\documentclass[twocolumn, journal]{IEEEtran}
\usepackage{graphics,color}
\usepackage{amsmath,amssymb,epsfig,dsfont,amsbsy}
\usepackage{verbatim,subfigure}

\newtheorem{theorem}{Theorem}

\newtheorem{claim}[theorem]{Claim}
\newtheorem{corollary}[theorem]{Corollary}
\newtheorem{remark}{Remark}

\begin{document}

\title{Effects of the Generation Size and Overlap \\on Throughput and Complexity \\ in Randomized Linear Network Coding}
\author{Yao~Li,~\IEEEmembership{Student Member,~IEEE,}
Emina~Soljanin,~\IEEEmembership{Senior Member,~IEEE,}
and~Predrag~Spasojev\'c,~\IEEEmembership{Member,~IEEE,}
\thanks{Manuscript received April 15, 2010; revised August~14,~2010 and November~5,~2010.
The material in this work was presented in part at the IEEE International Symposium on Information Theory (ISIT'10), Austin, Texas, June 2010,
and the IEEE International Symposium on Network Coding (NetCod'10), Toronto,
Canada, June 2010. This work was in part supported by
NSF CNS Grant No. 0721888.}
\thanks{Y.\ Li and P.\ Spasojevi\'c are with WINLAB, the Department of Electrical and Computer Engineering,
Rutgers University, North Brunswick, NJ 08902 USA
emails:~\{yaoli,spasojev\}@winlab.rutgers.edu.}
\thanks{E.\ Soljanin is with the Mathematics of Networking and Communication Department, Enabling Computing Technologies, Bell Laboratories, Alcatel-Lucent, Murray Hill, NJ 07974, email: emina@alcatel-lucent.com.}
 }

\maketitle

\begin{abstract}
To reduce computational complexity and delay in randomized network
coded content distribution, and for some other practical reasons,
coding is not performed simultaneously over all content blocks, but
over much smaller, possibly overlapping subsets of these blocks,
known as generations. A penalty of this strategy is throughput reduction.
To analyze the throughput loss, we model coding over generations
with random generation scheduling as a coupon collector's brotherhood problem.
This model enables us to derive the expected number
of coded packets needed for successful decoding of the entire
content as well as the probability of decoding failure (the latter
only when generations do not overlap) and further, to quantify the tradeoff
between computational complexity and throughput. Interestingly, with
a moderate increase in the generation size, throughput quickly
approaches link capacity. Overlaps between generations can further improve
throughput substantially for relatively small generation sizes.
\end{abstract}


\section{Introduction}
\subsection{Motivation: Coding over Disjoint and Overlapping Generations}

Random linear network coding was proposed in
\cite{traceyISIT2003} for ``robust, distributed transmission and
compression of information in networks''. Subsequently, the idea found
a place in a peer-to-peer(P2P) file distribution system
Avalanche \cite{avalanche} from Microsoft. In P2P systems such as
BitTorrent, content distribution involves fragmenting the content at
its source, and using swarming techniques to disseminate the
fragments among peers. Systems such as Avalanche, instead,
circulate linear combinations of content fragments, which can be
generated by any peer. The motivation behind such a scheme is that,
it is hard for peers to make optimal decisions on the scheduling of
fragments based on their limited local vision, whereas when
fragments are linearly combined at each node, topology diversity is
implanted inherently in the data flows and can be exploited without
further co-ordination.

The introduction of network coding in P2P content distribution
systems brings about the issue of computational complexity. Consider
distributing a file consisting of $N$ fragments, each made up of $d$
symbols from a Galois field $GF(q)$ of size $q$. It takes
$\mathcal{O}(Nd)$ operations in $GF(q)$ to form a linear combination
per coded packet, and $\mathcal{O}(N^3+N^2d)$ operations, or,
equivalently, $\mathcal{O}(N^2+Nd)$ operations per information
packet, to decode the information packets by solving linear
equations. According to the implementers of UUSee \cite{uusee}, a peer-to-peer video streaming application
using randomized linear coding, even with the most optimized implementation,
going beyond $512$ fragments in each generation risks taxing a low-end CPU, typically used in power-efficient notebook computers.

In an effort to reduce computational complexity, information packets are partitioned into disjoint
subsets referred to as {\em generations}, and coding is done only within generations.
This approach scales down the
encoding and decoding problem from the whole file size $N$
to the generation size times the number of generations. The concept of generation in network coding was first proposed by Chou et al.\ in
\cite{choupractical} to handle the issue of network
synchronization. Coding over randomly scheduled generations was first theoretically analyzed by Maymounkov et al.\ in \cite{petarchunked}. Random scheduling of generations provides
the ``rateless'' property which reduces the need for receiver feedback and
 offers resilience to various erasure patterns over the communication link. In addition, in the peer-to-peer content distribution setting, random scheduling is to some degree a good approximation when global co-ordination among peers is impractical.

With random scheduling of generations, coded packets accumulate faster in some generations
than in others, even if all generations are scheduled equally probably.
While waiting for the last generation to become decodable, redundant packets are accumulated in other
generations. The situation is aggravated as the generation size decreases.
 One way to recover some of the throughput loss due to random scheduling without losing the benefits of reduced
 generation sizes is to allow generations to help each other in decoding. If the generations are allowed
to overlap, after some of the ``faster'' generations
are decoded, the number of unknown variables can be reduced
in those generations sharing information packets with the decoded ones, which in turn reduces the number of
coded packets needed to decode those generations, and enhances the throughput as a result.
Our goal is to characterize the effects of generation size and overlaps on the throughput and
complexity in randomized linear network coding.


\subsection{Related Work}

The performance of codes with random scheduling of
disjoint generations was first theoretically analyzed in
\cite{petarchunked} by Maymounkov et al., who referred to them as
{\em chunked codes}. Chunked codes allow convenient encoding at
intermediate nodes, and are readily suitable for peer-to-peer file
dissemination. In \cite{petarchunked}, the authors used an
adversarial schedule as the network model and characterized the code
performance under certain restrictions on the chunk(generation) size when the length of the information to be encoded tends to infinity.

Coding with overlapping generations was first studied in
\cite{queensoverlap} and \cite{carletonoverlap} with the goal to
improve throughput. Reference \cite{carletonoverlap} studied a ``head-to-toe'' overlapping scheme in
which only contiguous generations overlap for a given number of
information packets, and analyzed its asymptotic performance over a line network when the length of information goes to infinity. Another overlapping scheme with a
grid structure was proposed in \cite{queensoverlap}, analyzed
for short lengths (e.g., $4$ generations) and simulated for practical lengths. When properly designed,
 these codes show improved performance over codes with disjoint generations. In our work, we offer an analysis of
coding over disjoint and overlapping generations for finite but practically long information lengths.

\subsection{Organization and Main Contribution}
In this work, coding with both disjoint and
overlapping generations together with {\em random generation
scheduling} is studied from a coupon collection \cite{feller} perspective. Previously existing results from the classical coupon collector's problem, along with our extensions,
 enable us to characterize the code performance with finite information lengths, from which the asymptotic code performance can further be deduced.

Section \ref{sec:gen_model} introduces the general model for coding over generations, disjoint or overlapping,
over a unicast (binary erasure) link, and characterizes the computational cost for encoding and
decoding.

Section \ref{sec:rank} derives several results concerning linear independence among coded packets from the same generation. Such results serve
to link coupon collection to the decoding of content that has been encoded into multiple generations. Included (Claim~\ref{thm:wait_ccdf}) is a very good upper
 bound on the distribution of the number of coded packets needed for a specific generation for successful decoding.

Section \ref{sec:coupon} introduces the coupon collector's brotherhood problem and its variations that can be used to model coding over generations. Probability generating functions (Theorems \ref{thm:nonuniform_genfunc} and \ref{thm:uniform_genfunc}) and moments (Corollaries \ref{thm:nonuniform_moments} and \ref{thm:uniform_expected}) of the number of samplings needed to collect multiple copies of distinct coupons are derived for the random sampling of a {\it finite} set of coupons in Section \ref{subsec:coupon_finite}. Relevant asymptotic results on expected values and probability distributions in existing literature are recapitulated in Section \ref{subsec:coupon_limit} for performance characterization of coding over generations in the later part of the work. The section is presented in the coupon collection language and is in itself of independent interest for general readers interested in coupon collecting problems.

In Sections \ref{sec:nonoverlapping} and \ref{sec:overlapping}, results from the previous two sections are combined to enable the analysis of the effects of generation size and overlaps on the decoding latency/throughput of coding over disjoint or overlapping generations.

Section \ref{sec:nonoverlapping} studies the effects of generation size on the code throughput over a BEC channel for coding over disjoint generations. Section \ref{subsec:disjoint_moments} characterizes the mean and variance of the decoding latency (the number of coded packets transmitted until successful decoding) for {\em finite information lengths}, and Section \ref{subsec:disjoint_pe} provides a lower bound for the probability of decoding failure. A large gain in throughput is observed when the generation size increases from $1$ to a few tens.

In Section \ref{sec:overlapping}, the {\em random annex code} is proposed as an effort to improve code throughput by allowing random overlaps among generations. Section \ref{subsec:overlap_algorithm} lists an algorithm providing precise estimation of the expected decoding latency of the random annex code. The algorithm is based on the analysis of the overlapping structure in Section \ref{subsec:overlap_structure} and the results from the extended collector's brotherhood in Section \ref{sec:coupon}. Section \ref{subsec:overlap_results} demonstrates the effects of overlap sizes on code throughput is shown through both numerical computation and simulations. One of our interesting observations is that overlaps between generations can provide a tradeoff between computational complexity and decoding latency. In addition, without increasing the generation size (and hence computational complexity), it is still possible to improve code throughput significantly by allowing overlaps between generations.

\section{Coding Over Generations: The General Model} \label{sec:gen_model}
In this section, we describe a general random coding scheme over
generations.
Generations do not have to be disjoint or of
equal size, and random scheduling of generations does not have to be uniform.
We describe the coding scheme over a unicast link.

\subsection{Forming Generations} \label{subsec:formgen}
The file being distributed $\mathcal{F}$
is represented as a set of $N$ information packets,
$p_1,p_2,\dots,p_N$. Each information packet is a $d$-dimensional column vector of information
symbols in Galois Field $GF(q)$. Generations are non-empty subsets
of $\mathcal{F}$.

Suppose that $n$ generations,
$G_1,G_2,\dots,G_n$, are formed s.t.\ $\mathcal{F}=\cup_{j=1}^n
G_j$. A coding scheme is said to be non-overlapping if the
generations are disjoint, i.e., $\forall i\ne j$, $G_i\cap
G_j=\emptyset$; otherwise, the scheme is said to be overlapping.
The size of each generation $G_j$ is denoted by $g_j$, and its
elements $p^{(j)}_1,p^{(j)}_2,\dots,p^{(j)}_{g_j}$. For convenience,
we will occasionally also use $G_j$ to denote the matrix with
columns $p^{(j)}_1,p^{(j)}_2,\dots,p^{(j)}_{g_j}$.

\subsection{Encoding}
In each transmission, the source first selects one of the $n$
generations at random. The probability of choosing generation $G_i$
is $\rho_i$, $\sum_{i=1}^{n}\rho_i=1$. Let
$\boldsymbol{\rho}=(\rho_1,\rho_2,\dots,\rho_n)$. Once generation
$G_j$ is chosen, the source chooses a coding vector $\mathbf{e}=[e_1,e_2,\dots,e_{g_j}]^T$,
with each of the $g_j$ components
chosen independently and equally probably from $GF(q)$. A new packet
$\bar{p}$ is then formed by linearly combining packets from $G_j$ by
$\mathbf{e}$: $\bar{p}=\sum_{i=1}^{g_j}e_ip^{(j)}_{i}=\mathbf{e}\cdot G_j$
($G_j$ here denotes a matrix).

The coded packet $\bar{p}$ is then sent over the communication link to
the receiver along with the coding vector $\mathbf{e}$ and the
generation index $j$.
Figure \ref{fig:sysmodel} shows a diagram of the communication between the source and the receiver.
The generations shown in this example are chosen to be disjoint, but this is not necessary.
\begin{figure}[htbp]
\centering
\includegraphics[scale=0.8]{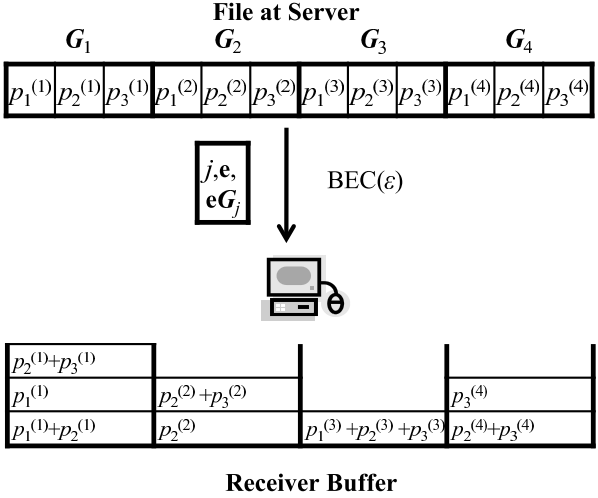}
\caption{A file divided into $N=12$ fragments and $n=4$ (disjoint) generations containing $h=3$ fragments each is available for distribution at the server. A receiver collects random linear combinations of randomly scheduled generations.}\label{fig:sysmodel}
\end{figure}


\subsection{Decoding}
Decoding starts with any generation $G_j$ for which the receiver has
collected $g_j$ coded packets with linearly independent coding
vectors. The information packets making up this generation are
decoded by solving a system of $g_j$ linear equations in $GF(q)$
formed by the coded packets on one side and the linear combinations
of the information packets by the coding vectors on the other. Since
generations are allowed to overlap, a decoded information packet
may also participate in other generations, from the equations of which the information packet is
then removed as an
unknown variable. Consequently,
in all the generations overlapping with the decoded generations,
the number of unknown packets is reduced. As a result, some generations may
become decodable even if no new coded packets are received from
the source. Again, the newly decoded generations resolve some
unknowns of the generations they overlap with, which in turn may
become decodable and so on. We declare successful decoding when all
$N$ information packets have been decoded.

The coding scheme described here is inherently rateless and easily
extendable to more general network topologies that allow coding at
intermediate network nodes.

\subsection{Packet Overhead}

Contained in each coded packet are the index of a generation $G_j$
and a linear combining vector for $G_j$ which together take up
$\lceil\log_2 n\rceil+g_j\lceil\log_2 q\rceil$ bits. Meanwhile, the
data in each coded packet comprise $d\lceil\log_2 q\rceil$ bits. The
generation size makes a more significant contribution to packet
overhead and such contribution is non-negligible due to the limited
size ($\sim$ a few KB) of transmission packets in practical
networks. This gives another reason to keep generations small,
besides reducing computational complexity.

\subsection{Computational Complexity}\label{subsec:complexity}
The computational complexity for encoding is
$\mathcal{O}(d\max\{g_j\})$ per coded packet for linearly combining
the $g_j$ information packets in each generation (recall that $d$ is
the number of $GF(q)$ symbols in each information packet, as defined
in Section \ref{subsec:formgen}). For decoding, the largest number
of unknowns in the systems of linear equations to be solved is not
more than $\max\{g_j\}$, and therefore the computational complexity
is upper bounded by $\mathcal{O}((\max \{g_j\})^2+d\max\{g_j\})$ per
information packet.

\subsection{Decoding Latency}\label{subsec:latency}
In this paper, we focus on the tradeoff between the computational
complexity and the {\it decoding latency} of these codes over
unicast links with erasures. Decoding latency here is defined as the
number of coded packets transmitted until successful decoding of all
the information packets, and {\it overhead} is the difference between the
number of information packets and the decoding latency. We assume a
memoryless BEC with a constant erasure rate $\epsilon$. Since our
coding scheme is rateless, each coded packet is statistically of the
same importance, and so the average decoding latency is inversely
proportional to the achievable capacity $(1-\epsilon)$ of the link. The throughput of the code
is inversely proportional to the decoding latency for given information length.




\section{Collecting Coded Packets and Decoding}\label{sec:rank}

A generation $G_i$ is not decodable until the number of
linearly independent equations collected for $G_i$ reaches the number
of its information packets not yet resolved by decoding other generations.
The connection between the number of coded
packets collected and the linear independence among these coded
packets has to be established before we can predict the decoding latency of codes over
generations using the collector's brotherhood model that will be
discussed in the next section.



Let $M(g,x)$ be the number of coded packets from a generation of size $g$
adequate for collecting $x$ linearly independent equations.
Then $M(g,x)$ has expected value \cite{monograph}

\begin{equation} \label{eq:hwait}
E[M(g,x)]=\sum_{j=0}^{x-1}\frac{1}{1-q^{j-g}}.
\end{equation}
Approximating summation by integration, from (\ref{eq:hwait}) we get
\begin{align}
E[M(g,x)]\lessapprox&\int_0^{x-1}\frac{1}{1-q^{y-g}}dy+\frac{1}{1-q^{x-1-g}}\notag\\
=&x+\frac{q^{x-1-g}}{1-q^{x-1-g}}+\log_q\frac{1-q^{-g}}{1-q^{x-1-g}}.
\end{align}
Let
\begin{equation}\label{eq:wait_eta}
\eta_g(x)=x+\frac{q^{x-1-g}}{1-q^{x-1-g}}+\log_q\frac{1-q^{-g}}{1-q^{x-1-g}}.
\end{equation}
We can use $\eta_g(x)$ to estimate the number of coded packets needed
from a certain generation to gather $x$ linearly independent
equations.

In addition, we have the following Claim \ref{thm:wait_ccdf} which upper bounds the tail probability of $M(g,g)$, the number of
coded packets needed for a certain generation to gather enough linearly independent equations for decoding.
\begin{claim}\label{thm:wait_ccdf}
There exist positive constants $\alpha_{q,g}$ and
$\alpha_{2,\infty}$ such that, for $s\ge g$,
\begin{align*}
&\textnormal{Prob}[M(g,g)>s] = 1-\prod_{k=0}^{g-1}(1-q^{k-s})\\
&<1-\exp(-\alpha_{q,g} q^{-(s-g)}) <1-\exp(-\alpha_{2,\infty}
q^{-(s-g)}).
\end{align*}
Also, since $1-\exp(-x)<x$ for $x>0$,
\begin{equation}
\textnormal{Prob}[M(g,g)>s]<\alpha_{q,g} q^{-(s-g)}.
\end{equation}
\end{claim}
\begin{proof}
Please refer to Appendix \ref{app:rank}.
\end{proof}

We will use Claim \ref{thm:wait_ccdf} in Theorem \ref{thm:disjoint_expected} in Section
\ref{sec:nonoverlapping} to derive an upper bound to the expected overhead of coding over disjoint
generations.


\section{Coupon Collector's Brotherhood and Collecting Coded Packets from Generations}
\label{sec:coupon}

The coupon collector's brotherhood
problem~\cite{doubledixiecup,brotherhood} studies quantities related
to the completion of $m$ sets of $n$ distinct coupons by sampling a
set of $n$ distinct coupons uniformly at random with replacement. In
analogy, coded packets belonging to generation $j$ can be
viewed as copies of coupon $j$, and hence the process of collecting
coded packets when generations are scheduled uniformly at random can
be modeled as collecting multiple copies of distinct coupons.

Because of possible linear dependence among coded packets and the overlaps between generations,
the numbers of coded packets needed
for each of the $n$ generations to ensure successful decoding, however, are $n$ random variables. Therefore, we must generalize the
coupon collector's brotherhood model from collecting a uniform number of copies for all coupons to collecting different numbers of copies for different coupons,
before it can be applied to the analysis
 of the throughput performance of
coding over generations. In this section, the original collector's brotherhood model is generalized in two ways. And later in this paper,
 the analysis of the throughput performance of coding over disjoint generations in Section~\ref{sec:nonoverlapping} rests on the first generalization,
 whereas that of coding over overlapping generations in Section~\ref{sec:overlapping} rests on the second generalization.
As our results are of more general interest than the coding-over-generations problem, we will express
them in the coupon collection language. For example, the probability $\rho_i$
of scheduling generation $G_i$ (defined in Section~\ref{sec:gen_model}) here
refers to the probability of sampling a copy of coupon $G_i$, for
$i=1,2,\dots,n$.

%



\subsection{Generating Functions, Expected Values and Variances}
\label{subsec:coupon_finite}
For any $m\in\mathbb{N}$, we define $S_m(x)$ as follows:
\begin{align}\label{eq:sm_m}
S_m(x)=&1+\frac{x}{1!}+\frac{x^2}{2!}+\dots+\frac{x^{m-1}}{(m-1)!}\quad(m\ge
1)\\
\label{eq:sm_0}S_{m}(x)=&0 \quad(m\le 0) ~\text{and} ~
S_{\infty}(x)= e^x.
\end{align}

Let the total number of samplings needed to ensure that at least
$m_i(\ge0)$ copies of coupon $G_i$ are collected for all
$i=1,2,\dots, n$ be $T(\boldsymbol{\rho},\mathbf{\mathbf{m}})$,
where $\mathbf{m}=(m_1,m_2,\dots,m_n)$. The following Theorem
\ref{thm:nonuniform_genfunc} gives $\varphi_{T(\boldsymbol{\rho},\mathbf{m})}(z)$,
 the generating function of the
tail probabilities of $T(\boldsymbol{\rho},\mathbf{\mathbf{m}})$. This result is
generalized from \cite{doubledixiecup} and \cite{brotherhood}, and
its proof uses the Newman-Shepp symbolic method in
\cite{doubledixiecup}. Boneh et al.~\cite{couponrevisited} gave the same
generalization, but we restate it here for use in our analysis of
coding over disjoint generations.  If for each $j=1,2,\dots,n$, the number of coded packets needed from generation $G_j$ for
its decoding is known to be $m_j$ (which can be strictly larger than the generation size $g_j$), $T(\boldsymbol{\rho},\mathbf{\mathbf{m}})$
then gives the total number of coded packets needed to ensure successful decoding of the entire content
when the generations are scheduled according to the probability vector $\boldsymbol{\rho}$.


\begin{theorem}(Non-Uniform Sampling)\label{thm:nonuniform_genfunc}
Let
\begin{equation}\label{eq:nonuniform_genfunc_def}
\varphi_{T(\boldsymbol{\rho},\mathbf{m})}(z)=\sum_{i\ge0}\mathrm{Prob}[T(\boldsymbol{\rho},\mathbf{m})>i]z^i.
\end{equation}
Then,
\begin{align}\label{eq:nonuniform_genfunc}
&\varphi_{T(\boldsymbol{\rho},\mathbf{m})}(z)=\\&\int_0^\infty
\Big\{e^{-x(1-z)}-
\prod_{i=1}^{n}\big[e^{-\rho_ix(1-z)}
-S_{m_i}(\rho_i
xz)e^{-\rho_i x}\big]\Big\}dx.\notag
\end{align}
\end{theorem}

\begin{proof}
 Please refer to Appendix \ref{app:coupon_proofs}, where we give a full
 proof of the theorem to demonstrate the Newman-Shepp symbolic method \cite{doubledixiecup},
 which is also used in the proof of our other generalization in Theorem~\ref{thm:uniform_genfunc}.
\end{proof}

The expected value and the variance of
$T(\boldsymbol{\rho},\mathbf{\mathbf{m}})$ follow from the tail
probability generating function derived in Theorem~\ref{thm:nonuniform_genfunc}.
\begin{corollary}\label{thm:nonuniform_moments}
\begin{align*}
E[T(\boldsymbol{\rho},\mathbf{\mathbf{m}})]&=\varphi_{T(\boldsymbol{\rho},\mathbf{m})}(1)\\
&=\int_0^\infty
\left\{1-\prod_{i=1}^{n}\left[1-S_{m_i}(\rho_i x)e^{-\rho_i
x}\right]\right\}dx,\\
Var[T(\boldsymbol{\rho},\mathbf{\mathbf{m}})]&=2\varphi_{T(\boldsymbol{\rho},\mathbf{m})}'(1)+\varphi_{T(\boldsymbol{\rho},\mathbf{m})}(1)
-\varphi_{T(\boldsymbol{\rho},\mathbf{m})}^2(1).
\end{align*}
\end{corollary}
\begin{proof}
Please refer to Appendix \ref{app:coupon_proofs}.
\end{proof}


Note that in Theorem~\ref{thm:nonuniform_genfunc} and Corollary~\ref{thm:nonuniform_moments}, $m_i$-s are allowed to be $0$, thus including the case where only a specific subset of the coupons is of interest.
Theorem~\ref{thm:nonuniform_genfunc} and Corollary~\ref{thm:nonuniform_moments} are also useful for the analysis of coding over generations when there is a difference in priority among the generations. For instance,
in layered coded multimedia content, the generations containing the packets of the basic layer could be
given a higher priority than those containing enhancement layers because of a hierarchical reconstruction at the receiver.

 In the following, we present another generalization of the collector's brotherhood model. Sometimes
 we are simply interested in collecting a coupon subset of a certain size,
 regardless of the specific content of the subset. This can be further extended to the following more complicated
 case: for each $i=1,2,\dots,A(A\ge1)$, ensure that there exists a subset of $\{G_1,G_2,\dots,G_n\}$ such that each of its $k_i$ elements has at least $m_i$ copies
 in the collected samples.
Such a generalization is intended for treatment of
coding over equally important generations, for example, when each generation is a
substream of multiple-description coded data. In this generalization, the generation scheduling (coupon sampling) probabilities are assumed to
be uniform, i.e., $\rho_1=\rho_2=\dots=\rho_n=1/n$.

Suppose that for some positive integer $A\le n$, integers $k_1,\dots,k_A$ and
$m_1,\dots,m_A$ satisfy $1\le k_1<\dots<k_A\le n$ and
$\infty=m_0>m_1>\dots> m_A>m_{A+1}=0$.
We are interested in the total number
$U(\mathbf{m},\mathbf{k})$  of coupons that needs to be collected,
to ensure that the number of
distinct coupons for which at least $m_i$ copies have been collected
is at least $k_i$, for all $i=1,2,\dots,A$, where
$\mathbf{m}=(m_1,m_2,\dots,m_A)$ and $\mathbf{k}=(k_1,k_2,\dots,k_A)$.
The following Theorem~\ref{thm:uniform_genfunc} gives the generating function $\varphi_{U(\mathbf{m},\mathbf{k})}(z)$ of $U(\mathbf{m},\mathbf{k})$.
\begin{theorem}(Uniform Sampling)\label{thm:uniform_genfunc}
\begin{align}\label{eq:uniform_genfunc}
&\varphi_{U(\mathbf{m},\mathbf{k})}(z)=n\int_{0}^{\infty}\!\!e^{-nx}\Bigl\{e^{nxz}-\\
&\!\!\!\!\sum_{{{(i_0,i_1,\dots,i_{A+1}):\atop i_0=0,i_{A+1}=n}\atop
i_j\in[k_j, i_{j+1}]}\atop j=1,2,\dots,A}\!\!\prod_{j=0}^{A}{{i_{j+1}}\!\!\choose{i_j}}\Bigl[S_{m_{j}}(xz)-S_{m_{j+1}}(xz)\Bigr]^{i_{j+1}-i_{j}}\Bigr\}dx.\notag
\end{align}
\end{theorem}
\begin{proof}
Please refer to Appendix \ref{app:coupon_proofs}.
\end{proof}

Same as for Corollary \ref{thm:nonuniform_moments}, we can find
$E[U(\mathbf{m},\mathbf{k})]=\varphi_{U(\mathbf{m},\mathbf{k})}(1)$.
A computationally wieldy representation of $E[U(\mathbf{m},\mathbf{k})]$ is offered in the following Corollary \ref{thm:uniform_expected}
in a recursive form.
\begin{corollary}\label{thm:uniform_expected}
For $k=k_{1}, k_{1}+1, \dots, n$, let
\begin{align*}
&\phi_{0,k}(x)=[(S_{m_0}(x)-S_{m_1}(x))e^{-x}]^{k};
\end{align*}
For $j=1,2,\dots,A$, let
\begin{align*}
&\phi_{j,k}(x)\\
&=\sum_{w=k_j}^{k}{k\choose
w}\left[(S_{m_{j}}(x)-S_{m_{j+1}}(x))e^{-x}\right]^{k-w}\phi_{j-1,w}(x),\\
&\textnormal{for }  k=k_{j+1},k_{j+1}+1,\dots,n.
\end{align*}
Then,
\begin{equation}
E[U(\mathbf{m},\mathbf{k})]=n\int_0^\infty\left(1-\phi_{A,n}(x)\right)dx.\label{eq:uniform_expected}
\end{equation}
\end{corollary}

It is not hard to find an algorithm that calculates
$1-\phi_{A,n}(x)$ in
$(c_1m_1+c_2(n-1)+c_3\sum_{j=1}^A\sum_{k=k_{j+1}}^n(k-k_j))$ basic
arithmetic operations, where $c_1$, $c_2$ and $c_3$ are positive
constants. As long as $m_1=\mathcal{O}(An^2)$, we can estimate the
amount of work for a single evaluation of $1-\phi_{A,n}(x)$ to be
$\mathcal{O}(An^2)$. The integral (\ref{eq:uniform_expected}) can be computed through the use of
an efficient quadrature method, for example, Gauss-Laguerre
quadrature. For reference, some numerical integration issues for the
special case where $A=1$ have been addressed in Part 7 of
\cite{Flajolet1992207} and in \cite{couponrevisited}.

In Section~\ref{sec:overlapping}, we will apply Corollary
\ref{thm:uniform_expected} to find out the expected throughput of the {\it random annex code},
an overlapping coding scheme in which generations share randomly
chosen information packets. The effect of the overlap size on the throughput can be investigated henceforth.  





\subsection{Limiting Mean Value and Distribution}\label{subsec:coupon_limit}
In the previous subsection, we considered collecting a finite number
of copies of a coupon set of a finite size. In this part, we present
some results from existing literature on the limiting behavior of
$T(\boldsymbol{\rho},\mathbf{m})$ as $n\rightarrow\infty$ or
$m_1=m_2=\dots=m_n=m\rightarrow\infty$, assuming
$\rho_1=\rho_2=\dots=\rho_n=\frac{1}{n}$. By slight abuse in
notation, we denote $T(\boldsymbol{\rho},\mathbf{m})$ here as
$T_n(m)$.

By Corollary \ref{thm:nonuniform_moments},
\begin{equation}
\label{eq:etnm} E[T_n(m)] =
n\int_{0}^{\infty}\left[1-(1-S_m(x)e^{-x})^n\right]dx.
\end{equation}

The asymptotics of $E[T_n(m)]$ for large $n$ has been discussed in
literature \cite{doubledixiecup}, \cite{flatto} and
\cite{couponnewaspects}, and is summarized in the following Theorem
\ref{thm:tnm_largen}, (\ref{eq:tnm_largem}), and Theorem
\ref{thm:tmn_dist}.

\begin{theorem}(\cite{flatto})\label{thm:tnm_largen}
When $n\rightarrow\infty$,
\begin{equation}\label{eq:tnm_largen}
E[T_n(m)]= n\log n+(m-1)n\log\log n+C_m n +o(n),
\end{equation}
 where $C_m=\gamma-\log(m-1)!$, $\gamma$
is Euler's constant, and $m\in\mathbb{N}$.
\end{theorem}

For $m\gg1$, on the other hand, we have \cite{doubledixiecup}
\begin{equation}\label{eq:tnm_largem}
E[T_n(m)]\rightarrow nm.
\end{equation}

What is worth mentioning is that, as the number of coupons
$n\rightarrow\infty$, for the first complete set of coupons, the
number of samplings needed is $\mathcal{O}(n\log n)$, whereas the
additional number of samplings needed for each additional set is
only $\mathcal{O}(n\log\log n)$.

In addition to the expected value of $T_n(m)$, the concentration of
$T_n(m)$ around its mean is also of great interest to us. This
concentration leads to an estimate of the probability of successful
decoding for a given number of collected coded packets. We can
specialize Corollary \ref{thm:nonuniform_moments} to derive the
variance of $T_n(m)$, as a measure of probability concentration.

Further, since the tail probability generating functions derived in
the last subsection are power series of non-negative coefficients
and are convergent at $1$, they are absolutely convergent on and
inside the circle $|z|=1$ in the complex $z$-plane. Thus, it is
possible to compute the tail probabilities using Cauchy's contour
integration formula. However, extra care is required for numerical
stability in such computation.

Here we instead look at the asymptotic case where the number of
coupons $n\rightarrow \infty$. Erd\"{o}s and R\'{e}nyi have proven
in \cite{renyi} the limit law of $T_n(m)$ as $n\rightarrow \infty$.
Here we restate Lemma B from \cite{flatto} by Flatto, which in
addition expresses the rate of convergence to the limit law. We will
later use this result to derive a lower bound for the probability of
decoding failure in Theorem \ref{thm:pe_lb} in Section~\ref{subsec:disjoint_pe}.

\begin{theorem}(\cite{flatto})\label{thm:tmn_dist}
Let \[Y_n(m)=\frac{1}{n}\left(T_n(m)-n\log n-(m-1)n\log\log
n\right).\] Then,
\[
\textnormal{Pr}[Y_n(m)\le y] =
\exp\left(-\frac{e^{-y}}{(m-1)!}\right)+\mathcal{O}\left(\frac{\log\log
n}{\log n}\right).
\]
\end{theorem}
\begin{remark}(Remarks 2\&3, \cite{flatto})\label{rmk:dist}
The estimation in Theorem \ref{thm:tmn_dist} is understood to hold
uniformly on any finite interval $-a\le y\le a$. i.e., for any
$a>0$,
\begin{align*}
\left|\textnormal{Prob}\left[Y_n(m)\le y\right]-\exp\left(-\frac{\exp(-y)}{(m-1)!}\right)\right|\le C(m,a)\frac{\log\log n}{\log n},&
\end{align*}
$n\ge 2$ and $-a\le y\le a$. $C(m,a)$ is a positive constant
depending on $m$ and $a$, but independent of $n$. For $m=1$, the
convergence rate to limit law is much faster: the
$\mathcal{O}\left(\frac{\log\log n}{\log n}\right)$ term becomes
$\mathcal{O}\left(\frac{\log n}{n}\right)$.
\end{remark}


\section{Coding Over Disjoint Generations}\label{sec:nonoverlapping}

In this section, we study the performance of coding over disjoint
generations. We derive both an upper bound and a lower bound for the
expected decoding latency (as defined in Section
\ref{subsec:latency}). We also derive the variance of the decoding
latency.

\subsection{Expected Decoding Latency and Its Variance}\label{subsec:disjoint_moments}
Let $M_i$ $(i=1,2,\dots,n)$ be the number of collected coded packets
from generation $G_i$ when $G_i$ first becomes decodable. Then
$M_i$ is at least $g_i$, has the same distribution as $M(g_i,g_i)$, the number of
coded packets needed for a certain generation to gather enough linearly independent equations for decoding, as defined and studied in Section~\ref{sec:rank}.
$M_i$'s are independent random variables. Let 
the decoding latency over a perfect channel be
$W(\boldsymbol{\rho},\mathbf{g}),$ where
$\mathbf{g}=(g_1,g_2,\dots,g_n).$ Use
$W_{\epsilon}(\boldsymbol{\rho},\mathbf{g})$ to denote the decoding
latency on a BEC($\epsilon$).

Let $X_k$ $(k=1,2,\dots)$ be i.i.d.\ geometric random variables with
success rate $1-\epsilon$. Therefore, $E[X_k]=\frac{1}{1-\epsilon}$ and
$E[X_k^2]=\frac{1+\epsilon}{(1-\epsilon)^2}$. Then
\[W_{\epsilon}(\boldsymbol{\rho},\mathbf{g})=\sum_{i=1}^{W(\boldsymbol{\rho},\mathbf{g})}X_i,\]
and therefore,
\begin{align}
E[W_{\epsilon}(\boldsymbol{\rho},\mathbf{g})]&=\frac{1}{1-\epsilon}E[W(\boldsymbol{\rho},\mathbf{g})],\label{eq:disjoint_epsilon_expected}\\
Var[W_{\epsilon}(\boldsymbol{\rho},\mathbf{g})]&=\frac{1}{(1-\epsilon)^2}\left(Var[W(\boldsymbol{\rho},\mathbf{g})]+\epsilon
E[W^2(\boldsymbol{\rho},\mathbf{g})]\right). \label{eq:disjoint_epsilon_variance}
\end{align}

By definition,
$E[W(\boldsymbol{\rho},\mathbf{g})]$ is lower bounded by
 $E[T(\boldsymbol{\rho},\mathbf{g})]$, the expected number of coded
 packets necessary for collecting at least $g_i$ coded packets for each generation $G_i$, and
$E[T(\boldsymbol{\rho},\mathbf{g})]$ is as given in Corollary~\ref{thm:nonuniform_moments}.

The following Theorem \ref{thm:disjoint_expected} gives the
exact expression for the first and second moments of
$W(\boldsymbol{\rho},\mathbf{g})$, along with an upper bound for
$E[W(\boldsymbol{\rho},\mathbf{g})]$ considering the effect of finite finite field size $q$. Then, the expected value and
the variance of $W_\epsilon(\boldsymbol{\rho},\mathbf{g})$ can be derived from (\ref{eq:disjoint_epsilon_expected}) and (\ref{eq:disjoint_epsilon_variance}).

\begin{theorem}\label{thm:disjoint_expected} The expected number of
coded packets needed for successful decoding of all $N$ information
packets

\begin{align}
E[&W(\boldsymbol{\rho},\mathbf{g})]\notag\\
=&
\int_0^{\infty}\left(1-\prod_{i=1}^{n}\left(1-e^{-\rho_ix}E_{M_i}\left[S_{M_i}(\rho_ix)\right]\right)\right)dx\\
<&\int_0^{\infty}\bigg(1-\prod_{i=1}^{n}\Big(1-e^{-\rho_ix}\big(S_{g_i}(\rho_ix) \label{eq:disjoint_expected_upper} \\
&+\alpha_{q,g_i}q^{g_i}e^{\rho_ix/q}-\alpha_{q,g_i}q^{g_i}S_{g_i}(\rho_ix/q)\big)\Big)\bigg)dx,\notag
\end{align}
\begin{align}
E[&W^2(\boldsymbol{\rho},\mathbf{g})] \label{eq:disjoint_sec_moment}\\
=& 2\int_0^{\infty}x\bigg(1
-\sum_{i=1}^{n}\rho_i\frac{1-E_{M_i}[S_{M_i-1}(\rho_i
x)]e^{-\rho_i x}}{1-E_{M_i}\left[S_{M_i}(\rho_i
x)\right]e^{-\rho_i
x}} \cdot\notag\\
&\cdot\prod_{j=1}^{n}\left(1-E_{M_j}\left[S_{M_j}(\rho_j
x)\right]e^{-\rho_j
x}\right)\bigg)dx   \notag\\
&+\int_0^{\infty}\left(1-\prod_{i=1}^{n}\left(1-e^{-\rho_ix}E_{M_i}[S_{M_i}(\rho_ix)]\right)\right)dx
\notag
\end{align}
where $\alpha_{q,g_i}=-\sum_{k=0}^{g_i-1}\ln\left(1-q^{k-g_i}\right),$ $i=1,2,\dots,n$.
\end{theorem}
\begin{proof}
Please refer to Appendix \ref{app:disjoint_expected}.
\end{proof}

In the case where generations are of equal size and scheduled
uniformly at random, we can estimate the asymptotic lower bound for
$E[W(\boldsymbol{\rho},\mathbf{g})]$ by the asymptotics of $T_n(m)$
given in (\ref{eq:tnm_largen}) and (\ref{eq:tnm_largem}).

\begin{figure}[htbp]
\begin{center}
\subfigure[]{\label{subfig:T}\includegraphics[scale=0.5]{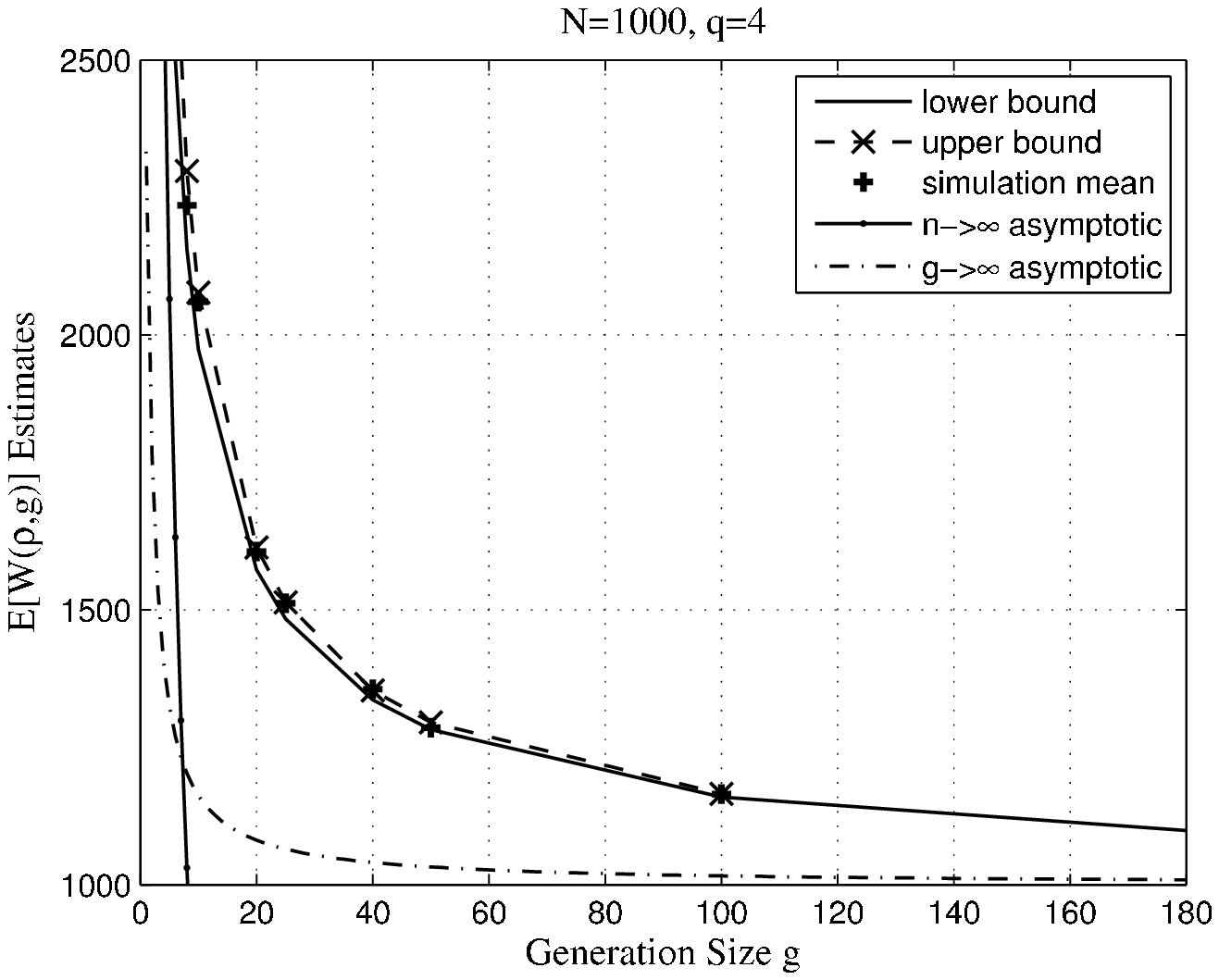}}\qquad
\subfigure[]{\label{subfig:std}\includegraphics[scale=0.5]{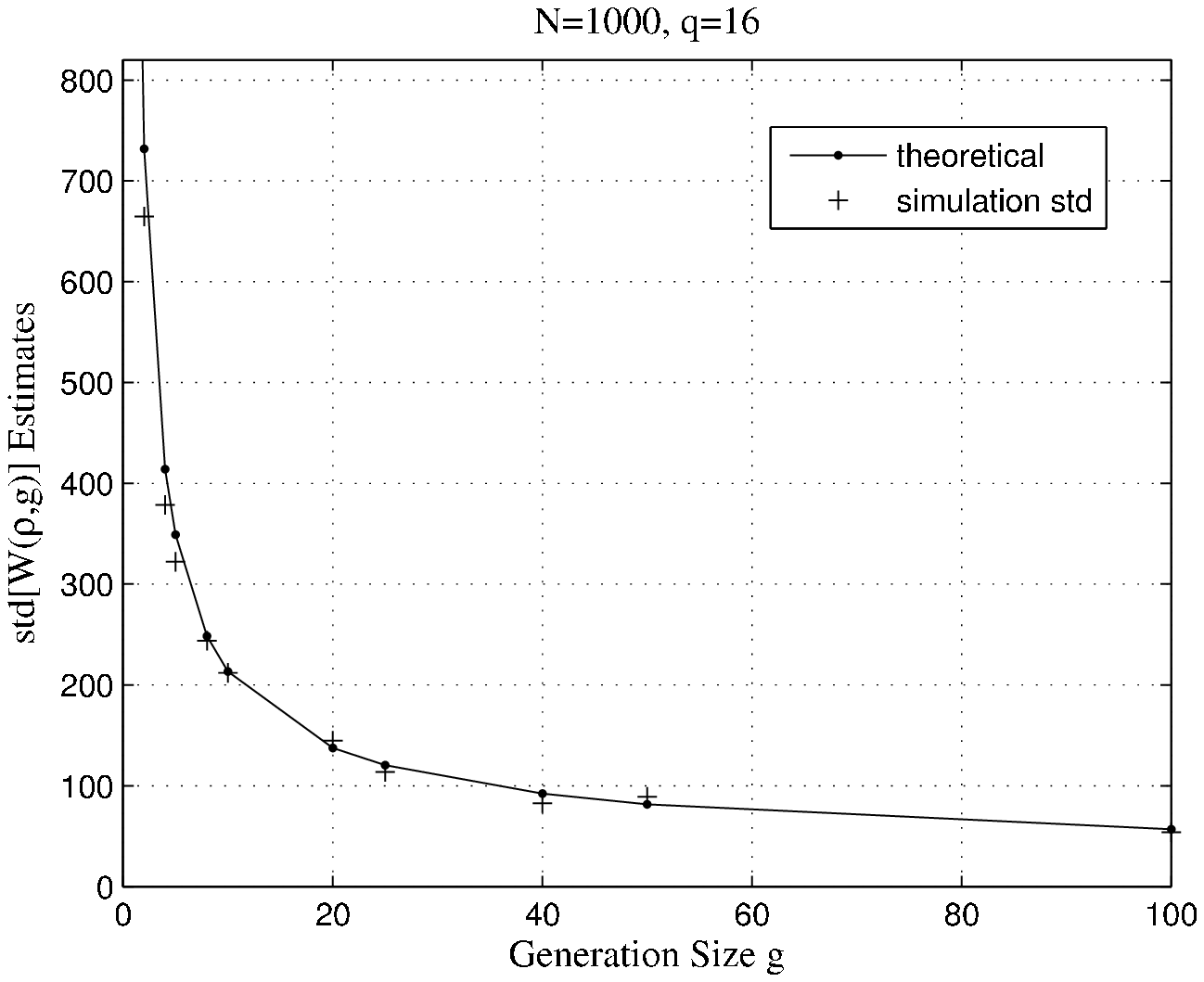}}\qquad
\subfigure[]{\label{subfig:pe}\includegraphics[scale=0.5]{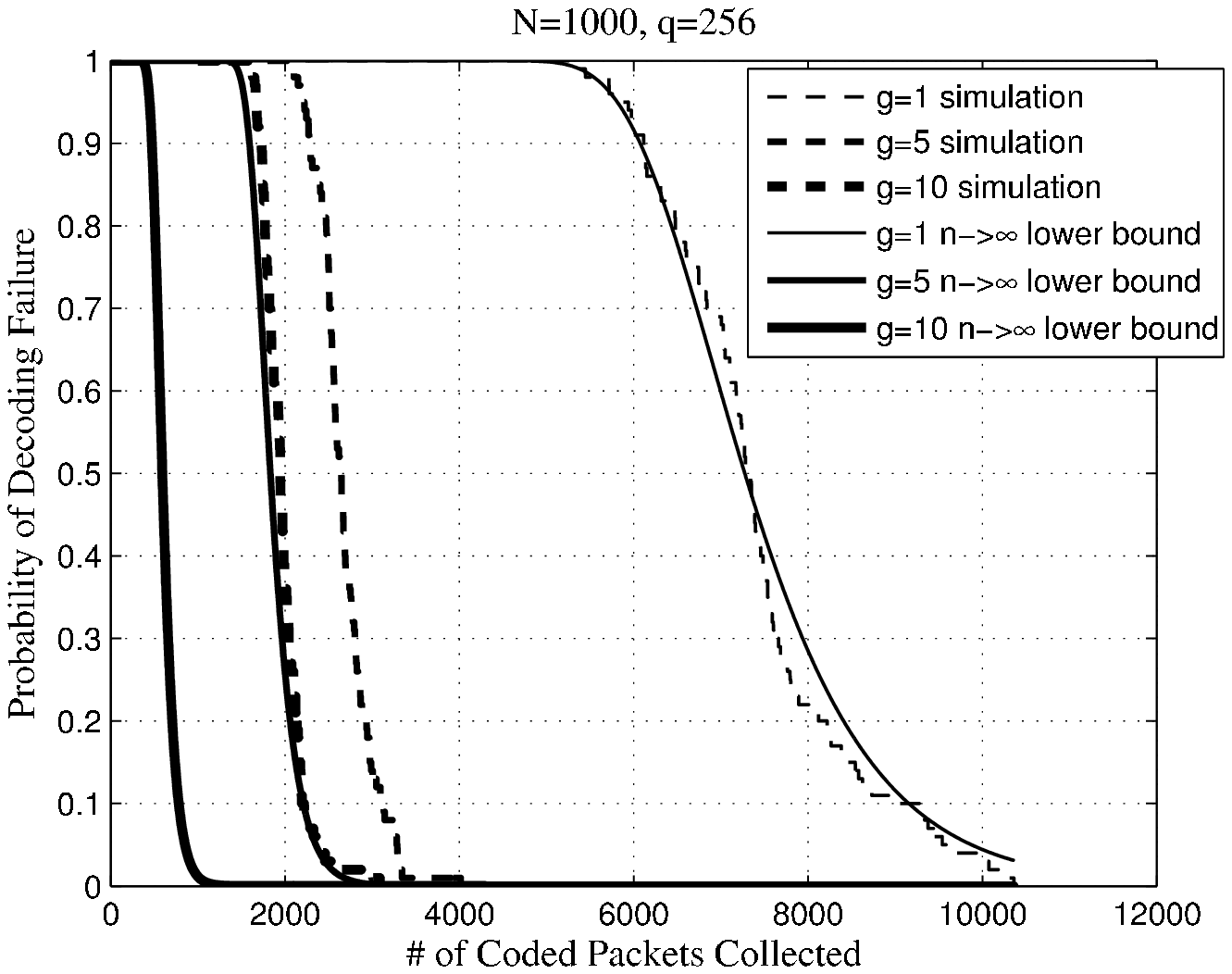}}
\caption{(a) Estimates of $E[W(\boldsymbol{\rho},\mathbf{g})]$, the
expected number of coded packets required for successful decoding
when the total number of information packets is $N=1000$, and both
$\mathbf{g}$ and $\boldsymbol{\rho}$ are uniform. Estimates shown:
lower bound $E[T(\boldsymbol{\rho},\mathbf{g})]$; upper bound
(\ref{eq:disjoint_expected_upper}); mean of
$W(\boldsymbol{\rho},\mathbf{g})$ in simulation;
$n\rightarrow\infty$ asymptotic (\ref{eq:tnm_largen}); $m\gg1$
asymptotics (\ref{eq:tnm_largem}); (b) Estimates of the standard
deviation of $W(\boldsymbol{\rho},\mathbf{g})$; (c) Estimates of
probability of decoding failure versus the number of coded packets
collected: Theorem \ref{thm:pe_lb} along with simulation results.}
\label{fig:brotherhood}
\end{center}
\end{figure}


Figure~\ref{fig:brotherhood}\subref{subfig:T} shows several
estimates of $E[W(\boldsymbol{\rho},\mathbf{g})]$, and Figure~\ref{fig:brotherhood}\subref{subfig:std} shows the standard deviation of $W(\boldsymbol{\rho},\mathbf{g})$ calculated from Theorem \ref{thm:disjoint_expected} and simulation results, when
$\rho_i=\frac{1}{n}$ and $g_i=g$ for $i=1,2,\dots,n$. The estimates
are plotted versus the uniform generation size $g$ for fixed
$N=ng=1000.$

For coding over disjoint generations and a fixed total number of information packets,
both the expected value and the standard deviation of the decoding latency 
drop significantly as the generation size $g$ grows to a relatively small
value from the case where no coding is used ($g=1$). Hence, throughput is improved by a moderate increase in the computational cost that scales quadratically with the generation size (see Section \ref{subsec:complexity}). On the other hand, we
also observe that past a moderate generation size ($\sim50-100$ coded packets for $N=1000$), the decrease in decoding latency becomes slower by further increasing the encoding/decoding complexity. We therefore argue for a ``sweet spot'' generation size which characterizes the tradeoff between throughput and complexity.





\subsection{Probability of Decoding
Failure}\label{subsec:disjoint_pe} In this subsection we assume
uniform generation size and scheduling probability, i.e.,
$\rho_i=\frac{1}{n}$, $g_i=g$ for $i=1,2,\dots,n$. For short, we
denote $W(\boldsymbol\rho,\mathbf{g})$ as $W_n(g)$. From Theorem
\ref{thm:tmn_dist}, we obtain the following lower bound to the
probability of decoding failure as $n\rightarrow\infty$:
\begin{theorem}\label{thm:pe_lb}
When $n\rightarrow\infty$, the probability of decoding failure when
$t$ coded packets have been collected is greater than
$1-\exp\left[-\frac{1}{(g-1)!}n(\log
n)^{g-1}\exp\left(-\frac{t}{n}\right)\right]+\mathcal{O}\left(\frac{\log\log
n}{\log n}\right)$.
\end{theorem}
\begin{proof}
The probability of decoding failure after acquiring $t$ coded
packets equals $\textnormal{Prob}[W_n(g)>t]$. Since $W_n(g)\ge
T_n(g)$,
\begin{align*}
&\textnormal{Prob}[W_n(g)>t] \ge \,\textnormal{Prob}[T_n(g)>t]\\
&=1-\!\textnormal{Prob}\left[Y_n(g)\le\frac{t}{n}-\log
n-(g-1)\log\log n\right].
\end{align*}
The result in Theorem \ref{thm:pe_lb} follows directly from Theorem
\ref{thm:tmn_dist}.
\end{proof}

\begin{corollary}\label{thm:overhead}
When $g$ is fixed and $n\rightarrow\infty$, in order to make the probability of decoding failure
smaller than $\delta$, the number of coded packets collected
has to be at least $E[T_n(g)]-n\log\log\frac{1}{1-\delta}$. If
$\delta=\frac{1}{N^c}$ for some constant $c$, then the number of
coded packets necessary for successful decoding has to be at least
$E[T_n(g)]+cn\log(ng)$.
\end{corollary}

Theorem 4.2 in \cite{petarchunked} also gives the number of coded packets
needed to have the probability of decoding failure below $\delta=\frac{1}{N^c}$, but under the assumption that
$\ln(N/\delta)=o(N/n)=o(g)$. In comparison, Corollary \ref{thm:overhead} treats the case where $g$ is constant.

Figure~\ref{fig:brotherhood}\subref{subfig:pe} shows the estimate of
the probability of decoding failure versus $T$, the number of coded
packets collected. As pointed out in Remark \ref{rmk:dist}, for
$m\ge2$, the deviation of the CDF of $T_n(m)$ from the limit law for
$n\rightarrow\infty$ depends on $m$ and is on the order of
$\mathcal{O}({\frac{\log\log n}{\log n}})$ for $m\ge2$, which is
quite slow, partly explaining the deviation of the limit law curves from the
simulation curves for $m=5$ and $m=10$ in Figure~\ref{fig:brotherhood}\subref{subfig:pe}.



\section{Coding Over Overlapping Generations}\label{sec:overlapping}


Even when generations are scheduled
uniformly at random, there will be more coded packets accumulated in
some of the generations than in others. The ``slowest'' generation is the bottleneck
for file decoding.
It is then advisable to design a
mechanism that allows ``faster'' generations to help those lagging
behind. In this section, we propose the {\it random annex code}, a
new coding scheme in which generations share randomly chosen
packets, as opposed to previously proposed ``head-to-toe'' overlapping scheme of
\cite{carletonoverlap}.

We provide a heuristic
analysis of the code throughput based on our results for the coupon collection model and
an examination of the overlapping structure. Previous work on coding over
overlapping generations, \cite{queensoverlap} and
\cite{carletonoverlap}, lacks accurate performance analysis for
information blocks of moderate finite lengths. On the other hand,
the computational effort needed to carry out our analysis scales
well with the length of information, and the performance predictions
coincide with simulation data. In addition, we find that our random
annex code outperforms the ``head-to-toe'' overlapping scheme of
\cite{carletonoverlap} over a unicast link.

In this section we conveniently assume that the coded packets are
sent over a perfect channel, since here we are interested in
comparing the performance of different rateless coding schemes.

\subsection{Forming Overlapping Generations}
We form $n$ overlapping generations out of a file with $N$ information packets in two steps as follows:
\begin{enumerate}
  \item Partition the file set $\mathcal{F}$ of $N$ packets into subsets $B_1,B_2, \dots, B_n$, each
   containing $h$ consecutive packets. These $n=N/h$
  subsets are referred to as \textit{base generations}.
Thus, $B_i=\{p_{(i-1)h+1},p_{(i-1)h+2},\dots,p_{ih}\}$ for
$i=1,2,\dots,n$. $N$ is assumed to be a multiple of $h$ for
convenience. In practice, if $N$ is not a multiple of $h$,  set
$n=\lceil N/h \rceil$ and assign the last $[N-(n-1)h]$ packets to
the last (smaller) base generation.
  \item To each base generation $B_i$, add a random \emph{annex} $R_i$, consisting of
$l$ packets chosen uniformly at random (without replacement) from
the $N-h=(n-1)h$ packets in $\mathcal{F}\backslash B_i$. The base
generation together with its annex constitutes the \emph{extended
generation} $G_i=B_i\cup R_i$, the size of which is $g=h+l$.
Throughout this paper, unless otherwise stated, the term
``generation'' will refer to ``extended generation'' whenever used
alone for overlapping generations.
\end{enumerate}

The generation scheduling probabilities are chosen to be uniform,
$\rho_1=\rho_2=\dots=\rho_n=1/n$. The encoding and decoding procedures run the
same as described in the general model in Section
\ref{sec:gen_model}.

\subsection{Analyzing the Overlapping Structure}\label{subsec:overlap_structure}
The following Claims \ref{thm:pi} through \ref{thm:overlap_gennum}
present combinatorial derivations of quantities concerning the frequency at which an arbitrary
information packet is represented in different generations.

\begin{claim}\label{thm:pi}
For any packet in a base generation $B_{k}$, the probability that it
belongs to annex $R_r$ for some $r\in\{1,2,\dots,n\}\backslash\{k\}$
is
\[\pi={{N-h-1}\choose{l-1}}/{{N-h}\choose{l}}=\frac{l}{N-h}=\frac{l}{(n-1)h},\]
whereas the probability that it does not belong to $R_r$ is $\bar{\pi}
= 1-\pi.$
\end{claim}

\begin{claim}\label{thm:representation}
Let $X$ be the random variable representing the number of
generations an information packet participates in. Then, $X=1+Y,$
where $Y$ is ${\rm Binom}(n-1,\pi).$
\[E[X]=1+(n-1)\pi=1+\frac{l}{h},\] and \[Var[X]=(n-1)\pi\bar{\pi}.\]
\end{claim}

\begin{claim} \label{thm:multiplerepresentation}
In each generation of size $g=h+l$, the expected number of information packets
not participating in any other generation is
$h\bar{\pi}^{(n-1)}\approx he^{-l/h}$ for $n\gg1$; the expected
number of information packets participating in at least two generations is \[l+h[1-\bar{\pi}^{(n-1)}]\approx
l+h\left[1-e^{-l/h}\right]<\min\{g,2l\}\] for $n\gg1$ and $l>0$.
\end{claim}

\begin{claim}\label{thm:overlap_gennum}
The probability that two generations overlap is
$1-{{N-2h}\choose {l,l,N-2h-2l}}/{{N-h}\choose l}^2$. The number of
generations overlapping with any one generation $G_i$ is then $${\rm
Binom}\left(n-1,\left[1-{{N-2h}\choose {l,l,N-2h-2l}}/{{N-h}\choose
l}^2\right]\right).$$
\end{claim}

The following Theorem \ref{thm:union_overlap} quantifies the
expected amount of help a generation may receive from previously
decoded generations in terms of common information packets. In the next
subsection, we use Corollary \ref{thm:uniform_expected}
and Theorem \ref{thm:union_overlap} for a heuristic analysis of the expected
throughput performance of the random annex code.
\begin{theorem} \label{thm:union_overlap}
For any $I\subset\{1,2,\dots,n\}$ with $|I|=s$, and any $j\in
\{1,2,\dots,n\}\backslash I$,
\begin{align}
\Omega(s)=E[|\left(\cup_{i\in I} G_i\right)\cap G_j|]&= g\cdot
\left[1-\bar{\pi}^s\right] +  s h \cdot \pi \bar{\pi}^s
\label{eq:union_overlap}\end{align}
 where $|B|$ denotes the cardinality of
set $B$. When $n\rightarrow\infty$, if
$\frac{l}{h}\rightarrow\alpha$ and $\frac{s}{n}\rightarrow\beta$,
and let $\omega(\beta)=\Omega(s)$, then $\omega(\beta)\rightarrow
h\left[(1+\alpha) \left(1 - e^{-\alpha\beta}\right) +\alpha\beta
e^{-\alpha\beta}\right]$.
\end{theorem}
\begin{proof}
Please refer to Appendix \ref{app:union_overlap}.
\end{proof}

\subsection{Expected Throughput Analysis: The Algorithm}\label{subsec:overlap_algorithm}
Given the overlapping structure, we next describe an analysis of the
expected number of coded packets a receiver needs to collect in
order to decode all $N$ information packets of $\mathcal{F}$ when
they are encoded by the random annex code. We base our analysis on
Theorem \ref{thm:union_overlap} above, Corollary
\ref{thm:uniform_expected} in Section \ref{sec:coupon}, and also
(\ref{eq:wait_eta}) in Section \ref{sec:rank}, and use the mean
value for every quantity involved.

By the time when $s$ $(s=0,1,\dots,n-1)$ generations have been decoded,
for any one of the remaining $(n-s)$ generations, on the average
$\Omega(s)$ of its participating information packets have been
decoded, or equivalently, $(g-\Omega(s))$ of them are not yet resolved. If for
any one of these remaining generations the receiver has collected enough
coded packets to decode
its unresolved packets, that generation becomes the $(s+1)$th
decoded; otherwise, if no such generation exists, decoding
fails.

The quantity $\eta_g(x)$ defined in (\ref{eq:wait_eta}) in Section~\ref{sec:rank}
estimates the number of coded packets from a generation of size $g$
adequate for collecting $x$ linearly independent equations.
By extending the domain of $\eta_g(x)$ from integers to real numbers,
we can estimate that the number of coded packets needed for the
$(s+1)$th decoded generation should exceed $m'_s =\lceil
\eta_g(g-\Omega(s))\rceil$. Since in the random annex code, all
generations are randomly scheduled with equal probability, for
successful decoding, we would like to have at least $m'_0$ coded
packets belonging to one of the generations, at least $m'_{1}$
belonging to another, and so on. Then Corollary~\ref{thm:uniform_expected} in
Section~\ref{sec:coupon} can be applied to estimate the total
number of coded packets needed to achieve these minimum requirements
for the numbers of coded packets.

The algorithm for our heuristic analysis is listed as follows:
\begin{enumerate}
\item Compute $\Omega(s-1)$ for $s=1,\dots,n$ using Theorem~\ref{thm:union_overlap};
\item Compute $m_{s}^{\prime}=\lceil \eta_g(g-\Omega(s-1))\rceil$ for
$s=1,2,\dots,n$ using (\ref{eq:wait_eta});
\item \label{enum:cat}Map $m_s^{\prime}$ $(s=1,2,\dots,n)$ into $A$ values $m_j$ $(j=1,2,\dots,A)$ so that
$m_j=m_{k_{j-1}+1}^{\prime}=m_{k_{j-1}+2}^{\prime}=\dots=m_{k_{j}}^{\prime}$,
for $j=1,2,\dots,A$, $k_0=0$ and $k_A=n$;
\item Evaluate (\ref{eq:uniform_expected}) in Corollary \ref{thm:uniform_expected} with
the $A$, $k_j$s, and $m_j$s obtained in Step \ref{enum:cat}), as an
estimate for the expected number of coded packets needed for
successful decoding.
\end{enumerate}
\begin{remark}
The above Step \ref{enum:cat}) is viable because $\Omega(s)$ is
nondecreasing in $s$,  $\eta_g(x)$ is non-decreasing in $x$ for
fixed $g$, and thus $m_s^{\prime}$ is non-increasing in $s$.
\end{remark}

Although our analysis is heuristic, we will see in the next section
that the estimate closely follows the simulated average performance
curve of the random annex coding scheme.

\subsection{Numerical Evaluation and Simulation Results}\label{subsec:overlap_results}
\subsubsection{Throughput vs.\ Complexity in Fixed Number of
Generations Schemes} Our goal here is to find out how the annex size
$l$ affects the decoding latency of the scheme with fixed base
generation size $h$ and the total number of information packets $N$
(and consequently, the number of generations $n$). Note that the generation size
$g=h+l$ affects the computational complexity of the scheme, and hence we are actually looking at
the tradeoff between throughput and complexity.

Figure \ref{fig:fixnh} shows both the analytical and simulation results
when the total number $N$ of information packets is $1000$ and the
base generation size $h$ is $25$. Figure~\ref{fig:fixnh}\subref{subfig:fixhdist}
shows $h+l-\Omega(s)$ for
$s=0,1,\dots,n$ with different annex sizes. Recall that $\Omega(s)$
is the expected size of the overlap of the union of $s$ generations
with any one of the leftover $n-s$ generations. After the decoding
of $s$ generations, for any generation not yet decoded, the expected
number of information packets that still need to be resolved is then
$h+l-\Omega(s)$. We observe that the $h+l-\Omega(s)$ curves start
from $h+l$ for $s=0$ and gradually descends, ending somewhere above
$h-l$, for $s=n-1$.

Recall that we measure throughput by decoding latency (Section~\ref{subsec:latency}).
Figure~\ref{fig:fixnh}\subref{subfig:fixhexp} shows the expected
performance of the random annex code, along with the performance of the head-to-toe
overlapping code and the non-overlapping code ($l=0$).
Figure~\ref{fig:fixnh}\subref{subfig:fixhpe} shows the probability of
decoding failure of these codes versus the number of coded packets
collected.


\begin{figure}[h]
\begin{center}
\subfigure[]{\label{subfig:fixhdist}\includegraphics[scale=0.45]{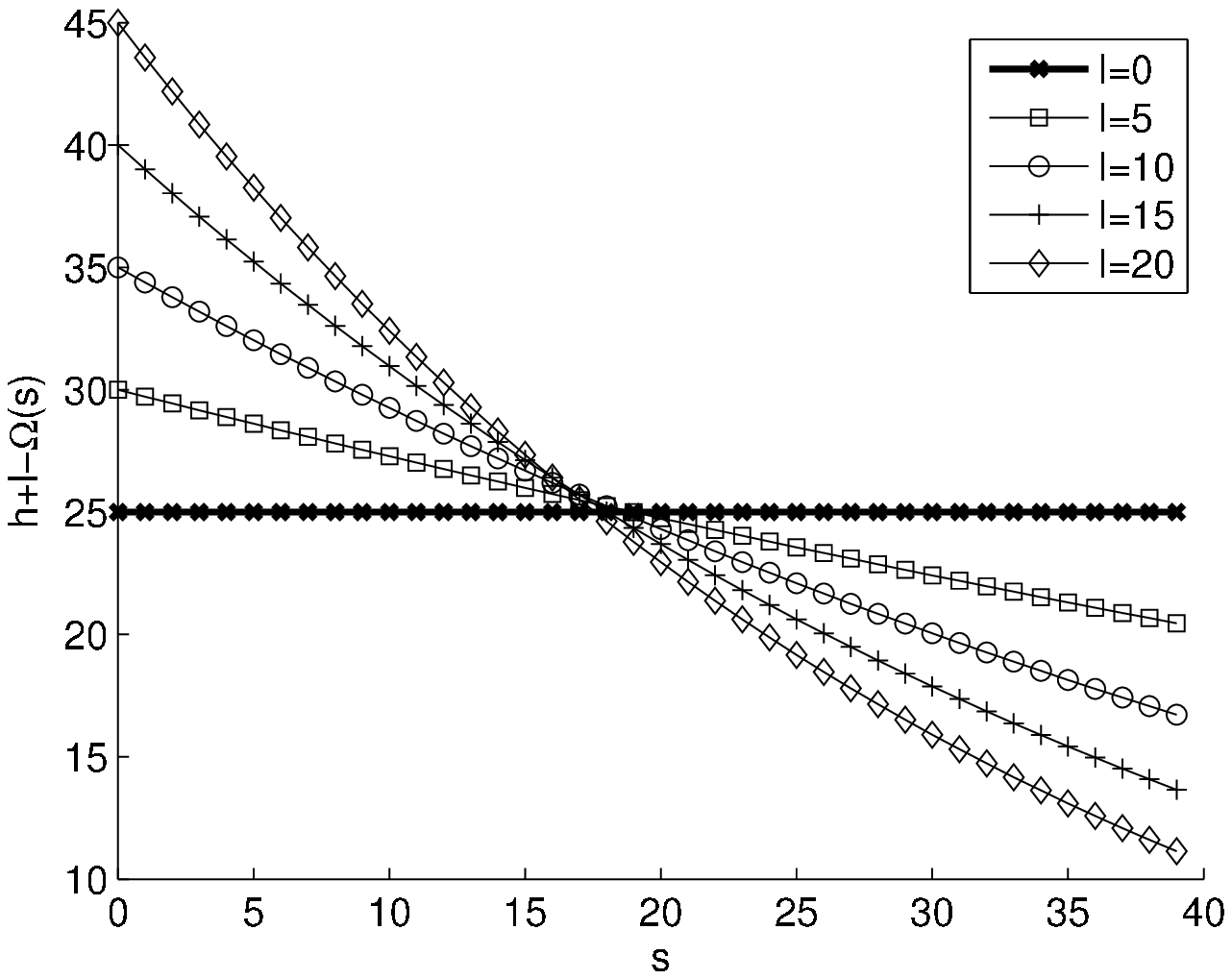}}\qquad
\subfigure[]{\label{subfig:fixhexp}\includegraphics[scale=0.45]{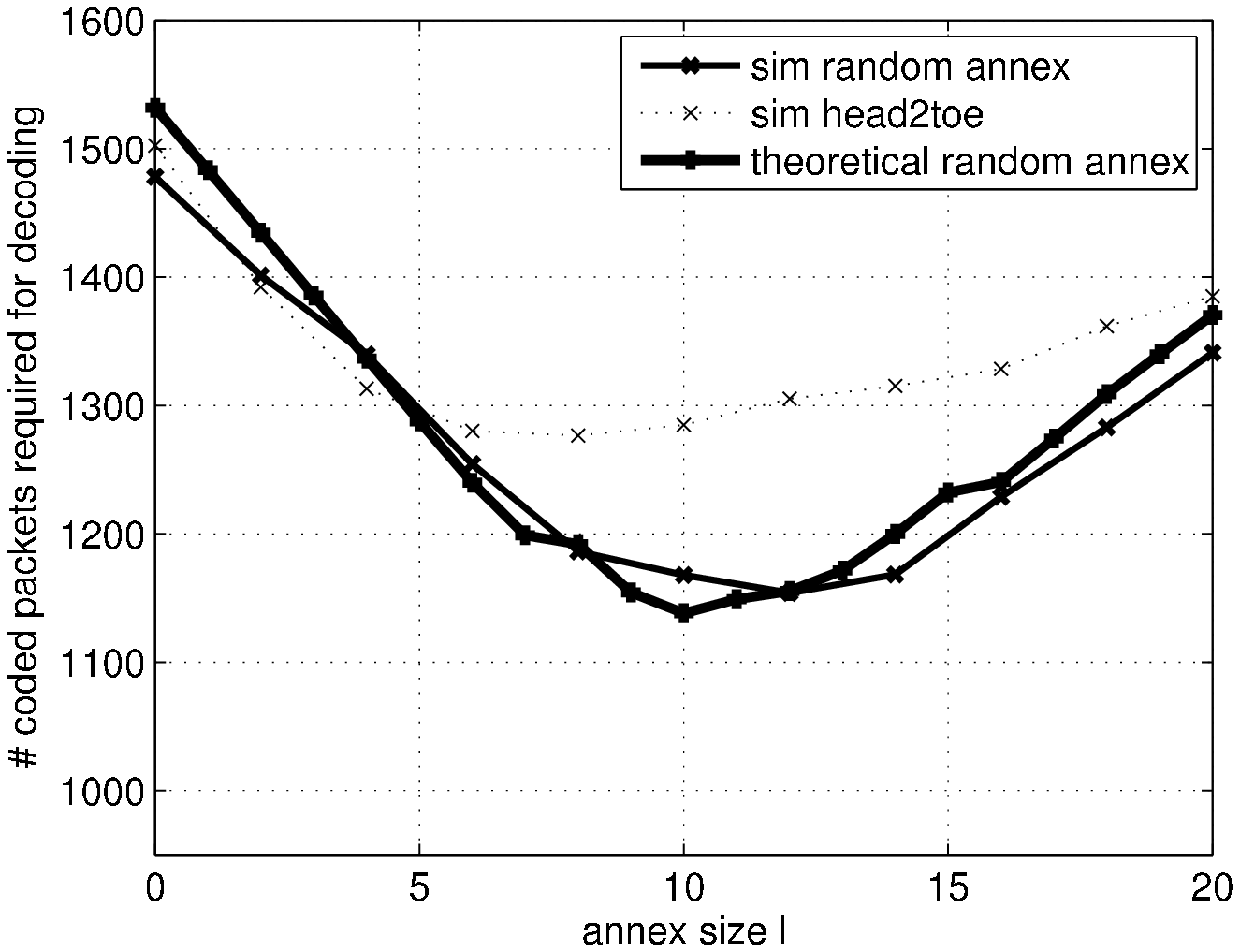}}\qquad
\subfigure[]{\label{subfig:fixhpe}\includegraphics[scale=0.45]{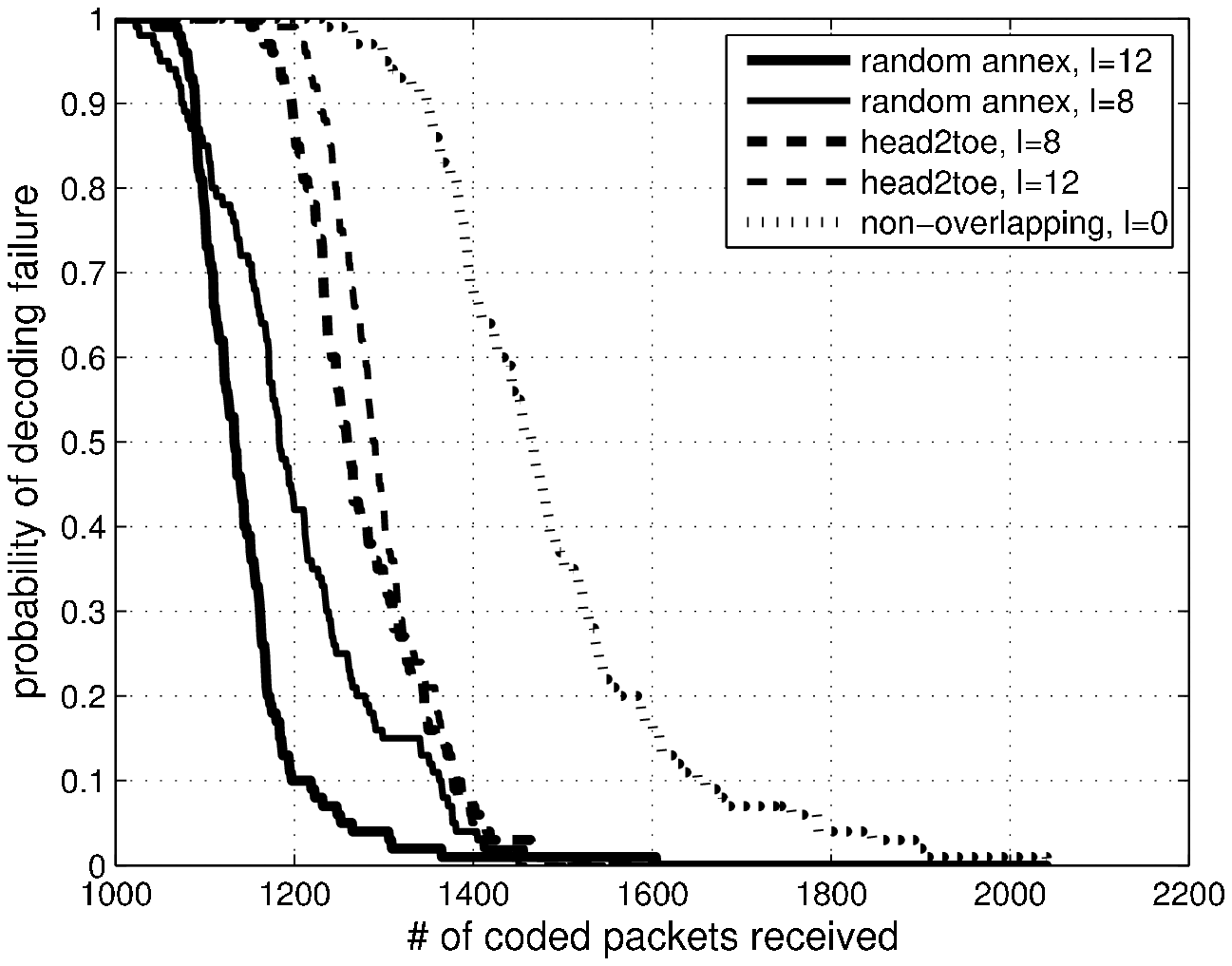}}
\caption{$N=1000$, $h=25$, $q=256$: \subref{subfig:fixhdist}
Difference between the generation size and the expected size of overlap with previously decoded generations $(h+l-\Omega(s))$;
\subref{subfig:fixhexp} Expected number of coded packets needed for
successful decoding versus annex size $l$; \subref{subfig:fixhpe}
Probability of decoding failure }\label{fig:fixnh}
\end{center}
\end{figure}

\begin{itemize}
\item Our analysis for the expected decoding latency closely matches the simulation
results.
\item Figure \ref{fig:fixnh}\subref{subfig:fixhexp} shows that by fixing the file size $N$ and the base generation size $h$, the expected decoding latency decreases roughly linearly with increasing annex size $l$,
up to $l=12$ for the random annex scheme and up to $l=8$
for the head-to-toe scheme. Meanwhile, the decoding cost per information packet
is quadratic in $g=h+l$. Beyond the optimal annex size,
throughput cannot be further increased by raising computational cost.
\item The random annex code outperforms head-to-toe overlapping at
their respective optimal points. Both codes outperform the
non-overlapping scheme.
\item As more coded packets are collected, the probability of decoding
failure of the random annex code converges to $0$ faster than that of the head-to-toe and that of the non-overlapping
scheme.
\end{itemize}

Overlaps provide a tradeoff between computational complexity and decoding latency.

\subsubsection{Enhancing Throughput in Fixed Complexity Schemes}
Our goal here is to see if we can choose the annex size to optimize
the throughput with negligible sacrifice in complexity. To this end,
we fix the extended generation size $g=h+l$ and vary only the annex
size $l$. Consequently, the computational complexity for coding does not increase when $l$ increases.
Actually, since some of the information packets in a generation of size $g$ could already be solved while decoding other generations, the remaining information packets in this generation can be solved in a system of linear equations of fewer than $g$ unknowns, and as a result increasing $l$ might decrease the decoding complexity.

Figure \ref{fig:fixgN} shows both the analytical and simulation results
for the code performance when the total number $N$ of information
packets is fixed at $1000$ and size $g$ of extended generation fixed
at $25$.
\begin{figure}[h]
\begin{center}
\subfigure[]{\label{subfig:fixgexp}\includegraphics[scale=0.45]{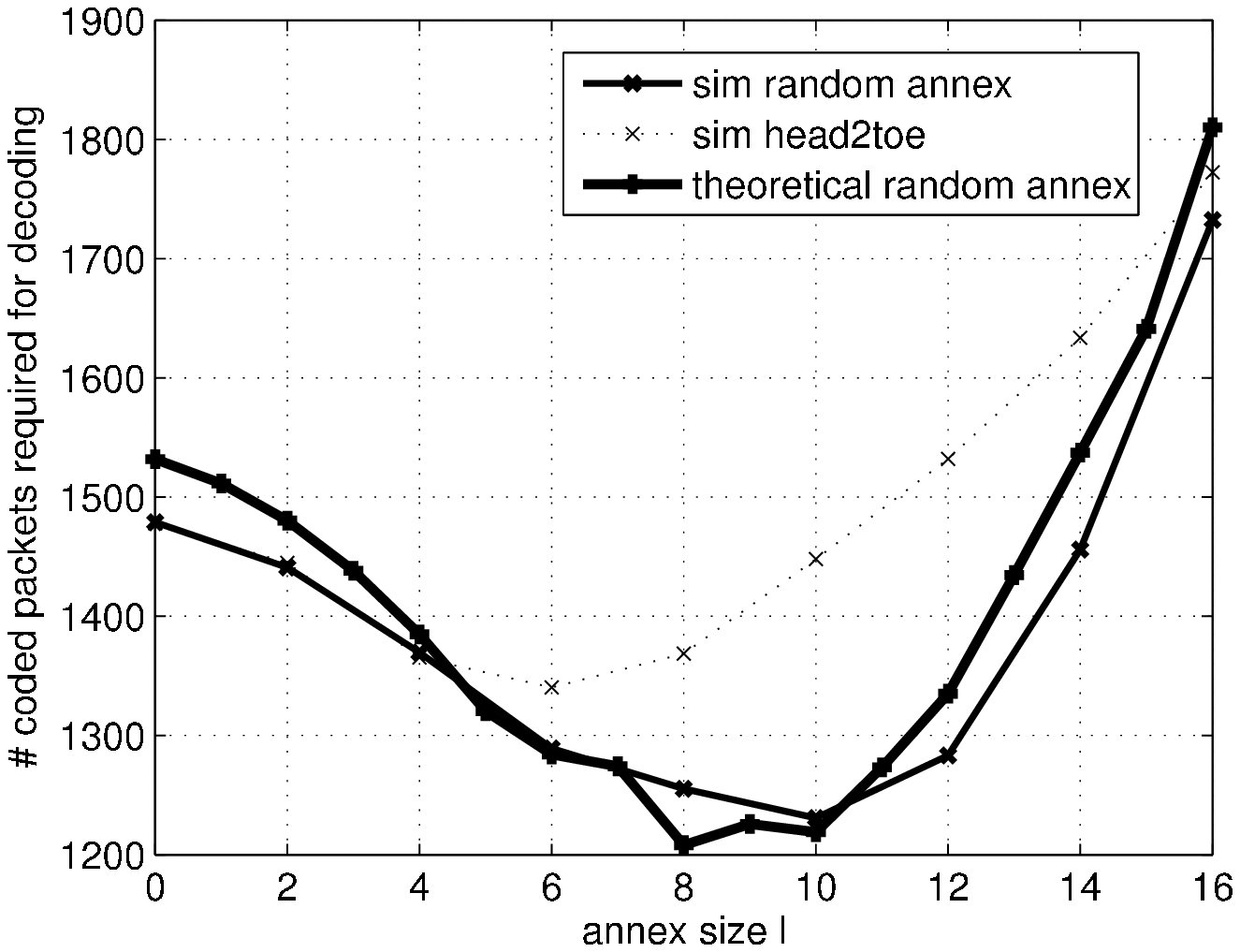}}\qquad
\subfigure[]{\label{subfig:fixgpe}\includegraphics[scale=0.45]{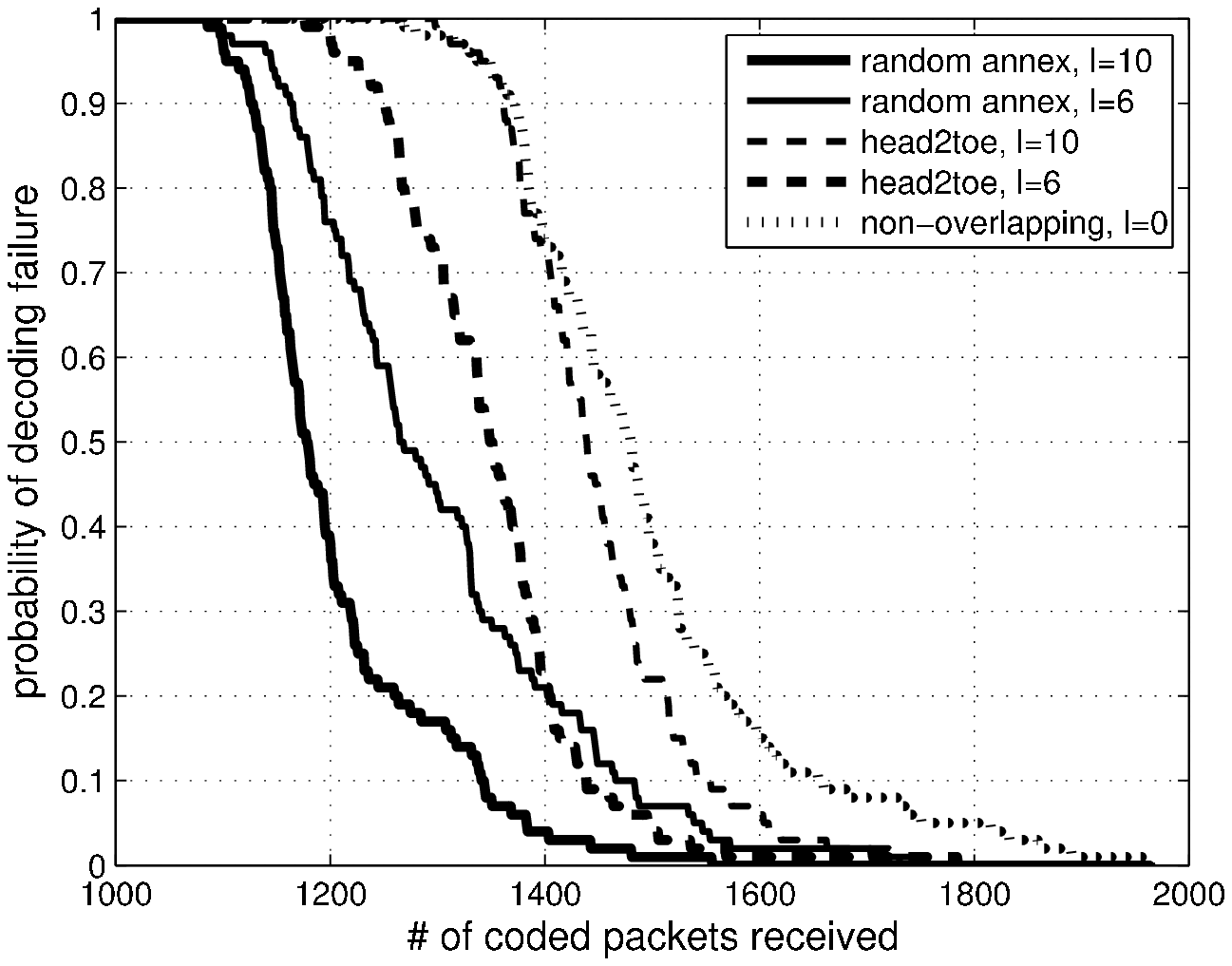}}
\caption{$N=1000$, $g=h+l=25$, $q=256$: \subref{subfig:fixgexp} Expected
number of coded packets needed for successful decoding versus annex
size $l$; \subref{subfig:fixgpe} Probability of decoding failure
}\label{fig:fixgN}
\end{center}
\end{figure}

\begin{itemize}
\item Again our analytical results agree with simulation results very well;
\item It is interesting to observe that, without raising computational complexity, increasing annex size properly can still give non-negligible improvement to throughput;
\item Figure
\ref{fig:fixgN}\subref{subfig:fixgexp} shows a roughly linear improvement of throughput with increasing $l$, up to $l=10$ for the random annex scheme and up to $l=6$ for the
head-to-toe scheme. Increasing $l$ beyond affects throughput adversely;
\item The random annex code again outperforms head-to-toe overlapping at
their optimal points. Both codes outperform the non-overlapping
scheme;
\item We again observe that the probability of decoding
failure of the random annex code converges faster than those of the
head-to-toe and the non-overlapping schemes.
\end{itemize}

When the overlap size increases, we either have larger generations with unchanged number of generations,
or a larger number of generations with unchanged generation size. In both cases the decoding latency would increase if we neglected the effect of overlaps during the decoding process. If we make use of the overlap in decoding, on the other hand, the larger the overlap size, the more help the generations can lend to each other in decoding and, hence, reducing the decoding latency. Two canceling effects result in a non-monotonic relationship between throughput and overlap size.

\begin{figure}[h]
\begin{center}
\includegraphics[scale=0.7]{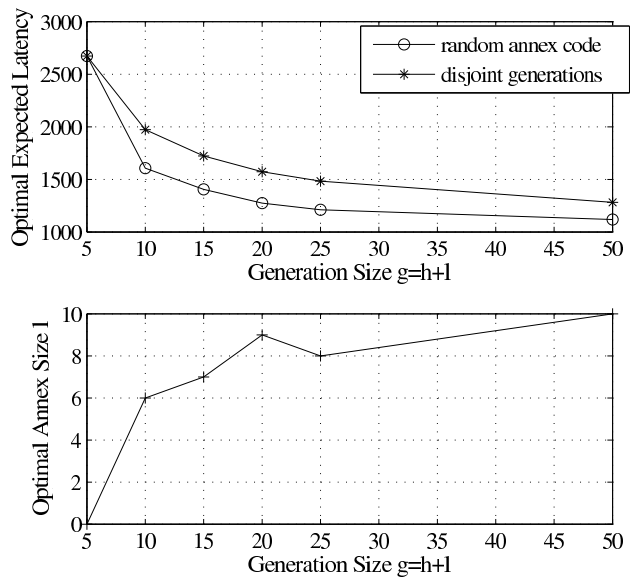}
\caption{Optimal expected decoding latency and the optimal overlap size with random annex codes. $N=1000$, $q=16$
}\label{fig:bestlat}
\end{center}
\end{figure}

The effect of generation size on the throughput of random annex codes is further illustrated in Figure \ref{fig:bestlat}. Figure \ref{fig:bestlat} plots the optimal expected decoding latency
achievable by random annex codes and the corresponding optimal annex size versus the generation size for $N=1000$ and $q=16.$ The plotted values are calculated
using the algorithm listed in Section \ref{subsec:overlap_algorithm}. We can see from Figure \ref{fig:bestlat} that with the random annex code and a generation size of $20$, the expected throughput is better than what can be achieved with coding over disjoint generations and a generation size of $50$. The reduction in computational complexity is considerable. Capturing the optimal overlap size in terms of other parameters of the code is our object of interest in the future.

\appendices

\section{Proof of Claim~\ref{thm:wait_ccdf}} \label{app:rank}\noindent For
$i=1,2,\dots,n$ and any $s\ge g$, we have
\begin{align*}
\ln\textnormal{Prob}&\bigl\{M(g,g)\le s\bigr\}   \\
 =&\ln\prod_{k=0}^{h-1}(1-q^{k-s})=\sum_{k=0}^{g-1}\ln(1-q^{k-s})\\
 =&-\sum_{k=0}^{g-1}\sum_{j=1}^{\infty}\frac{1}{j}q^{(k-s)j} =-\sum_{j=1}^{\infty}\frac{1}{j}\sum_{k=0}^{g-1}q^{j(k-s)}\\
 =&-\sum_{j=1}^{\infty}\frac{1}{j}q^{-js}\frac{q^{jg}-1}{q^{j}-1}\\
 =&-q^{-(s-g)}\sum_{j=1}^{\infty}\frac{1}{j}q^{-(j-1)(s-g)}\frac{1-q^{-jg}}{q^{j}-1}\\
 >&q^{-(s-g)}\sum_{j=1}^{\infty}\frac{1}{j}\frac{1-q^{-jg}}{1-q^{j}}\\
 =&q^{-(s-g)}\ln \textnormal{Prob}\bigl\{M(g,g)\le g\bigr\}\\
 >&q^{-(s-g)}\lim_{h\rightarrow\infty,q=2}\ln \textnormal{Prob}\bigl\{M(g,g)\le g\bigr\}
\end{align*}
The claim is obtained by setting \[ \alpha_{q,g}=-\ln
\textnormal{Prob}\bigl\{M(g,g)\le g\bigr\},\;\] and
\[\alpha_{2,\infty}=-\!\!\lim_{g\rightarrow\infty,q=2}\ln
\textnormal{Prob}\bigl\{M(g,g)\le g\bigr\}.
\]

\section{Proofs of Generalized Results of Collector's Brotherhood Problem}
\label{app:coupon_proofs}
\subsection*{Proof of Theorem~\ref{thm:nonuniform_genfunc}}\noindent
Our proof generalizes the symbolic method of \cite{doubledixiecup}.

Let $\xi$ 
be the event that the number of copies of coupon $G_i$ is at
least $m_i$ for every $i=1,2,\dots,n$. For integer $t\ge0$, let
$\xi(t)$ be the event that $\xi$ has occurred after a total of $t$
samplings, and let $\bar{\xi}(t)$ be the complementary event. Then,
the tail probability
$\mathrm{Prob}[T(\boldsymbol{\rho},\mathbf{m})>t]=\mathrm{Prob}[\bar{\xi}(t)]=\nu_t$.

To derive $\nu_t$, we introduce an operator $f$ acting on an
$n$-variable polynomial $g$. $f$ removes all monomials
$x_1^{w_1}x_2^{w_2}\dots x_n^{w_n}$ in $g$ satisfying $w_1\ge m_1,
\dots,w_n\ge m_n$. Note that $f$ is a linear operator, i.e., if
$g_1$ and $g_2$ are two polynomials in the same $n$ variables, and
$a$ and $b$ two scalars, we have $af(g_1)+bf(g_2)=f(ag_1+bg_2)$.

Each monomial in $(x_1+\dots+x_n)^t$ corresponds to one of the $n^t$
possible outcomes of $t$ samplings, with the exponent of $x_i$ being
the number of copies of coupon $G_i$. Since the samplings are
independent, the probability of an outcome $x_1^{w_1}x_2^{w_2}\dots
x_n^{w_n}$ is $\rho_1^{w_1}\rho_2^{w_2}\dots \rho_n^{w_n}$. Hence,
the probability of $\bar{\xi}(t)$ is $f((x_1+\dots+x_n)^t)$, when
evaluated at $x_i=\rho_i$ for $i=1,2,\dots n$, i.e.,
\begin{equation}\label{eq:nu_t}
\nu_t =f((x_1+\dots+x_n)^t)|_{x_i=\rho_i,
i=1,\dots,n}.\end{equation} Hence, (\ref{eq:nu_t}) and
(\ref{eq:nonuniform_genfunc_def}) lead to
\[\varphi_{T(\boldsymbol{\rho},\mathbf{m})}(z)
=\sum_{t\ge0} f\left((x_1+\dots+x_n)^t\right)z^t|_{x_i=\rho_i,
i=1,\dots,n}.\]

The identity
\[\int_{0}^{\infty}\frac{1}{t!}y^te^{-y}dy=1\] and
the linearity of the operator $f$ imply that
\begin{align}
\varphi_{T(\boldsymbol{\rho},\mathbf{m})}(z)
=&\int_{0}^{\infty}\sum_{t\ge0}\frac{f\left((x_1+\dots+x_n)^t\right)}{t!}z^ty^te^{-y}dy \notag\\
=&\int_{0}^{\infty}f\Bigl(\sum_{t\ge0}\frac{(x_1zy+\dots+x_nzy)^t}{t!}\Bigr)e^{-y}dy \notag\\
=&\int_{0}^{\infty}f\left(\exp(x_1zy+\dots+x_nzy)\right) e^{-y}dy
\label{eq:expsum}
\end{align}
evaluated at $x_i=\rho_i, i=1,\dots,n$.

We next find the sum of the monomials in the polynomial expansion of
$\exp(x_1+\dots+x_n)$ that should be removed under $f$. Clearly,
this sum should be $\prod_{i=1}^n
\left(e^{x_i}-S_{m_i}(x_i)\right)$, where $S$ is defined in
(\ref{eq:sm_m}) and (\ref{eq:sm_0})). Therefore,
\begin{align*}
&f\left(\exp(x_1zy+\dots+x_nzy)\right)|_{x_i=\rho_i, i=1,\dots,n}
\\&=e^{zy}-
\prod_{i=1}^n \left(e^{\rho_izy}-S_{m_i}(\rho_izy)\right).
\end{align*}
\begin{align}\label{eq:genfunc_b4_n}
\varphi_{T(\boldsymbol{\rho},\mathbf{m})}(z) =
&\int_{0}^{\infty}\left[e^{zy}-\prod_{i=1}^n
\left(e^{\rho_izy}-S_{m_i}(\rho_izy)\right)\right] e^{-y}dy
\end{align}

\subsection*{Proof of Corollary~\ref{thm:nonuniform_moments}}\noindent
Note that \begin{align*}
\varphi_{T(\boldsymbol{\rho},\mathbf{m})}(z)&=\sum_{t=0}^{\infty}\mathrm{Prob}[T(\boldsymbol{\rho},\mathbf{m})>t]z^t\\
&=\sum_{t=0}^{\infty}\sum_{j=t+1}^{\infty}\mathrm{Prob}[T(\boldsymbol{\rho},\mathbf{m})=j]z^t\\
&=\sum_{j=1}^{\infty}\mathrm{Prob}[T(\boldsymbol{\rho},\mathbf{m})=j]\sum_{t=0}^{j-1}z^t
\end{align*}
\begin{align*}
E[T(\boldsymbol{\rho},\mathbf{m})]&=\sum_{j=1}^{\infty}j\mathrm{Prob}[T(\boldsymbol{\rho},\mathbf{m})=j]=\varphi_{T(\boldsymbol{\rho},\mathbf{m})}(1).
\end{align*}
Similarly,
\begin{align*}
\varphi_{T(\boldsymbol{\rho},\mathbf{m})}'(z)&=\sum_{t=0}^{\infty}t\mathrm{Prob}[T(\boldsymbol{\rho},\mathbf{m})>t]z^{t-1}\\
&=\sum_{j=1}^{\infty}\mathrm{Prob}[T(\boldsymbol{\rho},\mathbf{m})=j]\sum_{t=0}^{j-1}tz^{t-1}\\
\varphi_{T(\boldsymbol{\rho},\mathbf{m})}'(1)&=\sum_{j=1}^{\infty}\frac{1}{2}j(j-1)\mathrm{Prob}[T(\boldsymbol{\rho},\mathbf{m})=j].
\end{align*}
Hence,
\begin{align*}
E[T(\boldsymbol{\rho},\mathbf{m})^2]
=&\sum_{j=1}^{\infty}j^2\mathrm{Prob}[T(\boldsymbol{\rho},\mathbf{m})=j]\\
=&2\varphi_{T(\boldsymbol{\rho},\mathbf{m})}'(1)+\varphi_{T(\boldsymbol{\rho},\mathbf{m})}(1),
\end{align*}
and consequently,
\begin{align*}
&\text{Var}[T(\boldsymbol{\rho},\mathbf{m})]=
2\varphi_{T(\boldsymbol{\rho},\mathbf{m})}'(1)+
\varphi_{T(\boldsymbol{\rho},\mathbf{m})}(1)-\varphi_{T(\boldsymbol{\rho},\mathbf{m})}^2(1).
\end{align*}

 We have
\begin{align*}
&\varphi_{T(\boldsymbol{\rho},\mathbf{m})}'(z)=\\
&\int_0^{\infty}x\biggl(e^{-x(1-z)} -\sum_{i=1}^{n}\rho_i\frac{
e^{-\rho_i x(1-z)}-S_{m_i-1}(\rho_i xz)e^{-\rho_i
x}}{e^{-\rho_ix(1-z)}-S_{m_i}(\rho_i xz)e^{-\rho_i
x}}\cdot\\
&\quad\cdot\prod_{j=1}^{n}\left(e^{-\rho_jx(1-z)}-S_{m_j}(\rho_j
xz)e^{-\rho_j x}\right)\biggr)dx,
\end{align*}
and from there, we can get $\varphi_{T(\boldsymbol{\rho},\mathbf{m})}'(1)$ and $\text{Var}[T(\boldsymbol{\rho},\mathbf{m})].$

\subsection*{Proof of Theorem~\ref{thm:uniform_genfunc}}\noindent
We again apply the Newman-Shepp symbolic method. Similar to the
proof of Theorem~\ref{thm:nonuniform_genfunc}, we introduce an
operator $f$ acting on an $n$-variable polynomial $g$. For a
monomial $x_1^{w_1}\dots x_n^{w_n}$, let $i_j$ be the number of
exponents $w_u$ among $w_1,\dots,w_n$ satisfying $w_u\ge k_j$, for
$j=1,\dots,A$. $f$ removes all monomials $x_1^{w_1}\dots x_n^{w_n}$
in $g$ satisfying $i_1\ge k_1, \dots,i_A\ge k_A$ and $i_1\le
\dots\le i_A$. $f$ is again a linear operator. One can see that
\begin{align}\label{eq:expsum}
&\varphi_{U(\mathbf{m},\mathbf{k})}(z)=\\
&\int_{0}^{\infty}f\left(\exp(x_1zy+\dots+x_nzy)\right)
e^{-y}dy|_{x_1=x_2=\dots=x_n=\frac{1}{n}}.\notag
\end{align}

We choose integers $0=i_0\le  i_1\le\dots\le i_A\le i_{A+1}=n$,
such that $i_{j}\ge k_{j}$ for $j=1,\dots,A$, and then partition
indices $\{1,\dots,n\}$ into $(A+1)$ subsets
$\mathcal{I}_1,\dots,\mathcal{I}_{A+1}$, where
$\mathcal{I}_j(j=1,\dots,A+1)$ has $i_{j}-i_{j-1}$ elements. Then
\begin{equation}\label{eq:expprodform}
\prod_{j=1}^{A+1}\prod_{i\in\mathcal{I}_{j}}(S_{m_{j-1}}(x_i)-S_{m_j}(x_i))
\end{equation}
equals the sum of all monomials in $\exp(x_1+\dots+x_n)$ with
($i_{j}-i_{j-1}$) of the $n$ exponents smaller than $m_{j-1}$ but
greater than or equal to $m_{j}$, for $j=1,\dots,A+1$. (Here $S$ is
as defined by (\ref{eq:sm_m}-{\ref{eq:sm_0}}).) The number of such
partitions of $\{1,\dots,n\}$ is equal to
${{n}\choose{n-i_A,\dots,i_2-i_1,i_1}}=\prod_{j=0}^{A}{{i_{j+1}}\choose
{i_{j}}}$. Finally, we need to sum the terms of the form
(\ref{eq:expprodform}) over all partitions of all choices of
$i_1,\dots,i_A$ satisfying $k_{j}\le i_{j}\le i_{j+1}$ for
$j=1,\dots,A$:
\begin{align}\label{eq:fform}
&f\left(\exp(x_1zy+\dots+x_nzy)\right)|_{x_1=\dots=x_n=\frac{1}{n}}
=\exp(zy)-\notag\\
&\sum_{{{(i_0,i_1,\dots,i_{A+1}):\atop i_0=0,i_{A+1}=n}\atop
i_j\in[k_j, i_{j+1}]}\atop
j=1,2,\dots,A}\prod_{j=0}^{A}{{i_{j+1}}\choose{i_j}}\left[S_{m_{j}}(\frac{zy}{n})-S_{m_{j+1}}(\frac{zy}{n})\right]^{i_{j+1}-i_{j}}.
\end{align}
Bringing (\ref{eq:fform}) into (\ref{eq:expsum}) gives our result in
Theorem \ref{thm:uniform_genfunc}.

\section{Proof of Theorem \ref{thm:disjoint_expected}}
\label{app:disjoint_expected}
\begin{align}
E[&W(\boldsymbol{\rho},\mathbf{g})] \notag\\
&=\sum_{\mathbf{m}}\left(\prod_{i=1}^n
\textnormal{Pr}[M_i=m_i]\right)
E[T(\boldsymbol{\rho},\mathbf{m})]\notag\\
&=\int_0^{\infty}\left[1-\prod_{i=1}^{n}\sum_{m_i}
\textnormal{Pr}[M_i=m_i](1-S_{m_i}(\rho_ix)e^{-\rho_ix})\right]dx \label{eq:disjoint_expected_exchange}\\
&=\int_0^{\infty}\left(1-\prod_{i=1}^{n}\left(1-e^{-\rho_ix}E_{M_i}\left[S_{M_i}(\rho_ix)\right]\right)\right)dx.\notag
\end{align}
(\ref{eq:disjoint_expected_exchange}) comes from the distributivity.

Since
\begin{equation*}
E_{M_i}\left[S_{M_i}(\rho_ix)\right]=\sum_{j=0}^{\infty}\frac{(\rho_ix)^j}{j!}\textnormal{Pr}[M_i>j],
\end{equation*}
by Claim~\ref{thm:wait_ccdf},
\begin{align*}
&E_{M_i}\left[S_{M_i}(\rho_ix)\right]\\
&<S_{g_i}(\rho_ix)+\sum_{j=g_i}^{\infty}\frac{(\rho_ix)^j}{j!}
 \alpha_{q,g} q^{-(j-g)}\\
 &=S_{g_i}(\rho_ix)+\alpha_{q,g_i}q^{g_i}e^{\rho_ix/q}-\alpha_{q,g_i}q^{g_i}S_{g_i}(\rho_ix/q),
\end{align*}
where $$\alpha_{q,g_i}=-\ln\textnormal{Pr}\bigl\{M(g_i,g_i)\le
g_i\bigr\}=-\sum_{k=0}^{g_i-1}\ln\left(1-q^{k-g_i}\right)$$
for $i=1,2,\dots,n$.

Hence, we have (\ref{eq:disjoint_expected_upper}).


Expression (\ref{eq:disjoint_sec_moment}) for
$E[W^2(\boldsymbol{\rho},\mathbf{g})]$ can be derived in the same manner, and then
the expression for $\text{Var}[W(\boldsymbol{\rho},\mathbf{g})]$ immediately follows.

\section{Proof of Theorem \ref{thm:union_overlap}}
\label{app:union_overlap}
 Without loss of
generality, let $I=\{1,2,\dots,s\}$ and $j=s+1$, and define
$\mathcal{R}_s = \cup_{i=1}^{s} R_i$, $\mathcal{B}_s=\cup_{i=1}^s
B_i$, and $\mathcal{G}_s = \cup_{i=1}^s G_i$ for $s=0,1,\dots,n-1$.
Then, $E\left[|\left(\cup_{i\in I} G_i\right)\cap
G_j|\right]=E\left[|\mathcal{G}_s \cap G_{s+1} |\right]$. For any
two sets $X$ and $Y$, we use $X+Y$ to denote $X\cup Y$ when $X\cap
Y=\emptyset$.
\begin{align*}
\mathcal{G}_s \cap G_{s+1} =& (\mathcal{B}_s + \mathcal{R}_s\backslash \mathcal{B}_s)\cap(B_{s+1}+R_{s+1})\\
=&\mathcal{B}_s\cap R_{s+1} + \mathcal{R}_s\cap B_{s+1} +
(\mathcal{R}_s\backslash \mathcal{B}_s)\cap R_{s+1},
\end{align*}
and therefore
\begin{align}\label{eq:part}
E[|\mathcal{G}_s \cap  G_{s+1}|]= & E[|\mathcal{B}_s\cap
R_{s+1}|]+\\ & E[|\mathcal{R}_s\cap B_{s+1}|] +
E[|(\mathcal{R}_s\backslash \mathcal{B}_s)\cap R_{s+1}|].\notag
\end{align}
Using Claim \ref{thm:pi}, we have
\begin{align}
&E[|\mathcal{B}_s\cap R_{s+1}|]=sh\pi,\label{eq:exp1}\\
&E[|\mathcal{R}_s\cap B_{s+1}|]=h[1-(1-\pi)^s],\label{eq:exp2}\\
&E[|(\mathcal{R}_s\backslash \mathcal{B}_s)\cap
R_{s+1}|]=(n-s-1)h\pi[1-(1-\pi)^s],\label{eq:exp3}
\end{align}
where $\pi$ is as defined in Claim \ref{thm:pi}.
 Bringing (\ref{eq:exp1})-(\ref{eq:exp3}) into (\ref{eq:part}),
we obtain (\ref{eq:union_overlap}).




Furthermore, when $n\rightarrow\infty$, if $l/h\rightarrow\alpha$
and $s/n\rightarrow\beta$, then
\begin{align*}
E[|\mathcal{G}_s \cap G_{s+1}|] =& g\cdot \left[1-\bar{\pi}^s\right]
+  s h \cdot \pi \bar{\pi}^s\\
\rightarrow&h(1+\alpha)
\Bigl[1-\Bigl(1-\frac{\alpha}{n-1}\Bigr)^{n\beta}\Bigr]+\\&h\alpha\beta
\Bigl(1-\frac{\alpha}{n-1}\Bigr)^{n\beta}\notag\\
\rightarrow& h\Bigl[(1+\alpha)(1-e^{-\alpha\beta})+\alpha
\beta e^{-\alpha\beta}\Bigr]\notag\\
=& h\Bigl[1+\alpha -
(1+\alpha-\alpha\beta)e^{-\alpha\beta}\Bigr]\notag
\end{align*}

\section*{Acknowledgement}
The authors would like to thank the anonymous reviewers for
helpful suggestions to improve the presentation of the paper.

\bibliographystyle{IEEEtran}

\begin{thebibliography}{10}
\providecommand{\url}[1]{#1}
\csname url@samestyle\endcsname
\providecommand{\newblock}{\relax}
\providecommand{\bibinfo}[2]{#2}
\providecommand{\BIBentrySTDinterwordspacing}{\spaceskip=0pt\relax}
\providecommand{\BIBentryALTinterwordstretchfactor}{4}
\providecommand{\BIBentryALTinterwordspacing}{\spaceskip=\fontdimen2\font plus
\BIBentryALTinterwordstretchfactor\fontdimen3\font minus
  \fontdimen4\font\relax}
\providecommand{\BIBforeignlanguage}[2]{{%
\expandafter\ifx\csname l@#1\endcsname\relax
\typeout{** WARNING: IEEEtran.bst: No hyphenation pattern has been}%
\typeout{** loaded for the language `#1'. Using the pattern for}%
\typeout{** the default language instead.}%
\else
\language=\csname l@#1\endcsname
\fi
#2}}
\providecommand{\BIBdecl}{\relax}
\BIBdecl

\bibitem{traceyISIT2003}
T.~Ho, R.~Koetter, M.~Medard, D.~Karger, and M.~Effros, ``The benefits of
  coding over routing in a randomized setting,'' in \emph{Proc. {IEEE}
  International Symposium on Information Theory ({ISIT}'03)}, 2003, p. 442.

\bibitem{avalanche}
C.~Gkantsidis and P.~Rodriguez, ``Network coding for large scale content
  distribution,'' in \emph{Proc. the 24th Annual Joint Conference of the IEEE
  Computer and Communications Societies (INFOCOM'05)}, vol.~4, Miami, FL, Mar.
  2005, pp. 2235--2245.

\bibitem{uusee}
Z.~Liu, C.~Wu, B.~Li, and S.~Zhao, ``{UUSee}: Large-scale operational on-demand
  streaming with random network coding,'' in \emph{Proc. the 30th {IEEE}
  Conference on Computer Communications ({INFOCOM}'10)}, San Diego, California,
  Mar. 2010.

\bibitem{choupractical}
P.~A. Chou, Y.~Wu, and K.~Jain, ``Practical network coding,'' in \emph{Proc.
  41st Annual Allerton Conference on Communication, Control, and Computing},
  Monticello, IL, Oct. 2003.

\bibitem{petarchunked}
P.~Maymounkov, N.~Harvey, and D.~S. Lun, ``Methods for efficient network
  coding,'' in \emph{Proc. 44th Annual Allerton Conference on Communication,
  Control, and Computing}, Monticello, IL, Sept. 2006.

\bibitem{queensoverlap}
D.~Silva, W.~Zeng, and F.~Kschischang, ``Sparse network coding with overlapping
  classes,'' in \emph{Proc. Workshop on Network Coding, Theory, and
  Applications (NetCod '09)}, Lausanne, Switzerland, Jun. 2009, pp. 74--79.

\bibitem{carletonoverlap}
A.~Heidarzadeh and A.~Banihashemi, ``Overlapped chunked network coding,'' in
  \emph{IEEE Information Theory Workshop (ITW'10)}, Jan. 2010, pp. 1--5.

\bibitem{feller}
W.~Feller, \emph{An Introduction to Probability Theory and Its Applications},
  3rd~ed.\hskip 1em plus 0.5em minus 0.4em\relax New York: John Wiley \& Sons
  Inc., 1968, vol.~1, pp. 224--225.

\bibitem{monograph}
C.~Fragouli and E.~Soljanin, \emph{{Network Coding Applications}}, ser.
  Foundations and Trends in Networking.\hskip 1em plus 0.5em minus 0.4em\relax
  Hanover, MA: now Publishers Inc., Jan. 2008, vol.~2, no.~2, pp. 146--147.

\bibitem{doubledixiecup}
D.~Newman and L.~Shepp, ``The double dixie cup problem,'' \emph{The American
  Mathematical Monthly}, vol.~67, no.~1, pp. 58--61, Jan. 1960.

\bibitem{brotherhood}
D.~Foata and D.~Zeilberger, ``The collector's brotherhood problem using the
  {Newman-Shepp} symbolic method,'' \emph{Algebra Universalis}, vol.~49, no.~4,
  pp. 387--395, 2003.

\bibitem{couponrevisited}
\BIBentryALTinterwordspacing
A.~Boneh and M.~Hofri, ``The coupon-collector problem revisited {--} a survey
  of engineering problems and computational methods,'' \emph{Stochastic
  Models}, vol.~13, pp. 39 -- 66, 1997. [Online]. Available:
  \url{http://www.informaworld.com/10.1080/15326349708807412}
\BIBentrySTDinterwordspacing

\bibitem{Flajolet1992207}
P.~Flajolet, D.~Gardy, and L.~Thimonier, ``Birthday paradox, coupon collectors,
  caching algorithms and self-organizing search,'' \emph{Discrete Applied
  Mathematics}, vol.~39, no.~3, pp. 207--229, 1992.

\bibitem{flatto}
\BIBentryALTinterwordspacing
L.~Flatto, ``Limit theorems for some random variables associated with urn
  models,'' \emph{The Annals of Probability}, vol.~10, no.~4, pp. 927--934,
  1982. [Online]. Available: \url{http://www.jstor.org/stable/2243548}
\BIBentrySTDinterwordspacing

\bibitem{couponnewaspects}
A.~N. Myers and H.~S. Wilf, ``Some new aspects of the coupon collector's
  problem,'' \emph{SIAM Rev.}, vol.~48, no.~3, pp. 549--565, 2006.

\bibitem{renyi}
P.~Erd{\"{o}}s and A.~R{\'e}nyi, ``On a classical problem of probability
  theory,'' \emph{Magyar Tud. Akad. Mat. Kutat\'o Int. K\"ozl.}, vol.~6, pp.
  215--220, 1961.

\end{thebibliography}

\end{document}